\documentclass[twocolumn, aps, superscriptaddress]{revtex4}

\usepackage{amsthm}
\usepackage{amsmath}
\usepackage{amssymb}
\usepackage{braket}
\usepackage{bbold}
\usepackage{graphicx}
\usepackage{subfigure}
\usepackage{hyperref}
\usepackage{appendix}
\newtheorem{theorem}{Theorem}
\newtheorem{lemma}{Lemma}
\newtheorem{corollary}{Corollary}

\theoremstyle{definition}

\hypersetup{colorlinks=true, citecolor=blue, urlcolor=blue, linkcolor=blue}

\begin{document}
\title{A new entanglement measure based on the total concurrence}

\author{Dong-Ping Xuan}
%\email{2230501014@cnu.edu.cn}
\affiliation{School of Mathematical Sciences, Capital Normal University, Beijing 100048, China}
\author{Zhong-Xi Shen}
%\email{18738951378@163.com}
\affiliation{School of Mathematical Sciences, Capital Normal University, Beijing 100048, China}
\author{Wen Zhou}
%\email{2230501027@cnu.edu.cn}
\affiliation{School of Mathematical Sciences, Capital Normal University, Beijing 100048, China}
\author{Zhi-Xi Wang}
\email{wangzhx@cnu.edu.cn}
\affiliation{School of Mathematical Sciences, Capital Normal University, Beijing 100048, China}
\author{Shao-Ming Fei}
\email{feishm@cnu.edu.cn}
\affiliation{School of Mathematical Sciences, Capital Normal University, Beijing 100048, China}
%\affiliation{Max-Planck-Institute for Mathematics in the Sciences, 04103, Leipzig, Germany}

\begin{abstract}
Quantum entanglement is a crucial resource in quantum information processing, advancing quantum technologies. The greater the uncertainty in subsystems' pure states, the stronger the quantum entanglement between them.
From the dual form of $q$-concurrence ($q\geq 2$) we introduce the total concurrence. A  bona fide measure of quantum entanglement is  introduced, the $\mathcal{C}^{t}_q$-concurrence ($q \geq 2$), which is based on the total concurrence. Analytical lower bounds for the $\mathcal{C}^{t}_q$-concurrence are derived. In addition,  an analytical expression is derived  for the $\mathcal{C}^{t}_q$-concurrence in the cases of isotropic and Werner states. Furthermore, the monogamy relations that the $\mathcal{C}^{t}_q$-concurrence satisfies for qubit systems  are  examined. Additionally, based on the parameterized $\alpha$-concurrence and its complementary dual,  the $\mathcal{C}^{t}_\alpha$-concurrence $(0\leq\alpha\leq\frac{1}{2})$ is  also proposed.
\medskip
%\noindent Keywords: Quantum entanglement; Total concurrence; Parameterized entanglement measure
\end{abstract}

\maketitle

\section{Introduction}

%Quantum entanglement [1¨C6] is an essential feature of quantum mechanics, which distinguishes the quantum from the classical world. As one of the fundamental differences between quantum entanglement and classical correlations, a key property of entanglement is that a quantum system entangled with one of the other systems limits its entanglement with the remaining others.

Quantum entanglement is a quintessential manifestation of quantum mechanics, which reveals the fundamental insights into the nature of quantum correlations \cite{Nielsen2000book,Horodecki865}. The quantification of quantum entanglement is a key aspect of entanglement measures from the perspective of resource theory of entanglement \cite{Vedral2275}. Various entanglement measures have been proposed for bipartite systems, including the concurrence introduced by Hill and Wootters based on linear entropy \cite{Hill5022}, negativity \cite{Zyczkowski883}, entanglement of formation (EoF)\cite{Bennett3824}, R\'{e}nyi-$\alpha$ entropy entanglement \cite{HORODECKI377,Gour012108}, Tsallis-$q$ entropy entanglement \cite{LANDSBERG211,Kim062328}, robustness entanglement \cite{Simon052327}, Unified-$(q,s)$ entropy entanglement \cite{Kim295303}, and so on.

Recently, there has been significant research on parameterized concurrence-based entanglement measures.
In \cite{Yang052423}, X. Yang et al. introduced a parameterized entanglement monotone called the $q$-concurrence ($q \geq 2$), which is closely associated with the general Tsallis-$q$ entropy.
Subsequently, Wei et al. \cite{WEI210} improved the analytic lower bound of this measure. Motivated by the Tsallis-$q$ entropy entanglement and the parameterized $q$-concurrence, Wei and Fei \cite{WEI275303} proposed a new parameterized entanglement measure, the $\alpha$-concurrence, which is defined for arbitrary $\alpha \in [0, 1/2]$. These entanglement measures are both defined solely by the reduced state of one subsystem. This limitation prompts a fundamental question about what additional insights can be gained from a given state beyond the information contained in the subsystems.
To address this issue, we introduce a novel parameterized concurrence in this paper, namely the $\mathcal{C}^{t}_q$-concurrence ($q \geq 2$).

This paper is structured as follows.
In Sec.~\ref{sect2}, we begin by defining the total concurrence through the $q$-concurrence and its complementary dual. Building on this, we introduce a new parameterized entanglement measure, the $\mathcal{C}^{t}_q$-concurrence ($q \geq 2$), which is derived from the total concurrence. In Sec.~\ref{sect3}, we present rigorous analytical lower bounds for the parameterized $\mathcal{C}^{t}_q$-concurrence ($q \geq 2$), ensuring their tightness and applicability in various quantum systems.
In Sec.~\ref{sect4}, we provide explicit analytical expressions for the $\mathcal{C}^{t}_q$-concurrence ($q \geq 2$) in the context of both isotropic and Werner states.
And we investigate the monogamy properties of the $C^{t}$-concurrence in multipartite qubit systems and higher-dimensional quantum states in Sec.~\ref{sect5}.
Finally, in Sec.~\ref{sect6}, we conclude by summarizing the key results and highlighting the significant findings of our study.

\section{Entanglement measure derived from the total concurrence of parameterized $q$-concurrence and its complementary dual}\label{sect2}

%In this section, we first present that  the complementary dual form of linear entropy and define the total linear entropy $E^{t}$ that contains more complete information, and provide some properties of total linear entropy.  Based on the known relationship between linear entropy and concurrence,  we define  $C_t$-concurrence entanglement based on concurrence and its complementary dual.  Finally, we prove that the defined $C_t$-concurrence entanglement is an bipartite entanglement measure and provided an analytical formula  for any two-qubit state.

Let $\mathcal{H}_A$ and $\mathcal{H}_B$ denote $d$-dimensional Hilbert spaces. The concurrence of a bipartite pure state $|\varphi\rangle_{A B} \in \mathcal{H}_A \otimes \mathcal{H}_B$ is defined as follows, as described in \cite{Rungta042315}:
$$
C\left(|\varphi\rangle_{A B}\right) = \sqrt{2\left(1 - \operatorname{tr} \varrho_A^2\right)},
$$
here, $\varrho_A = \operatorname{tr}_B\left(|\varphi\rangle_{A B} \langle \varphi|\right)$ represents the reduced density matrix of subsystem $A$, obtained by tracing out the degrees of freedom of subsystem $B$.
For a mixed state $\varrho_{AB}$, the concurrence is defined through its convex roof extension,
\begin{eqnarray}
C\left(\varrho_{A B}\right)=\inf _{\left\{p_i,\left|\varphi_i\right\rangle\right\}} \sum_i p_i C\left(\left|\varphi_i\right\rangle_{A B}\right)\label{concur},
\end{eqnarray}
where the infimum is taken over all decompositions of $\varrho_{A B} = \sum_i p_i \left|\varphi_i\right\rangle_{A B} \langle \varphi_i|$, with $p_i \geq 0$ and $\sum_i p_i = 1$.

In general, a well-defined quantum entanglement measure $\mathcal{E}$  should satisfy the following fundamental conditions: \cite{Horodecki3,Guhne1}:
\begin{itemize}
\item[(E1)] $\mathcal{E}(\rho) \geq 0$ for any state $\rho \in \mathcal{D}$, where the equality holds only for separable states.
\item[(E2)] $\mathcal{E}$ is invariant under local unitary transformations, $\mathcal{E}(\rho)=\mathcal{E}$ $\left(V_A \otimes V_B \rho V_A^{\dagger} \otimes V_B^{\dagger}\right)$ for any local unitaries $V_A$ and $V_B$.
\item[(E3)]$\mathcal{E}$ does not increase on average under  local operation and classical communication (LOCC),
$$
\mathcal{E}(\rho) \geq \sum_i p_i \mathcal{E}\left(\rho_i\right),
$$
here, $p_i=\operatorname{tr} A_i \rho A_i^{\dagger}$  represents the probability of obtaining the outcome $i$, while $\rho_i=A_i \rho A_i^{\dagger} / p_i$ is the corresponding post-measurement state. The operators $A_i$ are the Kraus operators associated with the  LOCC  process, satisfying the completeness relation $\sum_i A_i^{\dagger} A_i=I$.
\item[(E4)] $\mathcal{E}$ is convex,
$$
\mathcal{E}\left(\sum_i p_i \rho_i\right) \leq \sum_i p_i \mathcal{E}\left(\rho_i\right).
$$
\item[(E5)]  $\mathcal{E}$ cannot increase under $\operatorname{LOCC}, \mathcal{E}(\rho) \geq \mathcal{E}(\Lambda(\rho))$ for any LOCC map $\Lambda$.
\end{itemize}
The condition (E2) becomes unnecessary if the condition (E5) is satisfied. An entanglement measure $\mathcal{E}$ is classified as an entanglement monotone \cite{Vidal355} if it satisfies the first four conditions.
%In quantum entanglement resource theory, the free operations capture the physical scenario where spatially separated parties freely exchange classical information, but all quantum information is processed locally through completely positive and trace-preserving (CPTP) maps on the individual subsystems. The global maps that can be implemented under this restriction constitute the class of LOCC, and this represents the free operations in quantum entanglement resource.  Every LOCC operation is built from an interactive protocol in which each round of the protocol involves a local measurement by one of the parties followed by a global broadcast of the measurement outcome. By concatenating the Kraus operators for each round, it is then not too difficult to see that every LOCC map $\Lambda$ will have the form
%\begin{eqnarray}\label{LOCC}\Lambda(\cdot)=\sum_k\left(\otimes_{i=1}^N M_{k, i}^{A_i}\right)(\cdot)\left(\otimes_{i=1}^N M_{k, i}^{A_i}\right)^{\dagger}\end{eqnarray}where $M_{k, i}^{A_i}$ acts on the Hilbert space of party $A_i$. In other words, $\Lambda$ has a Kraus operator decomposition in which each Kraus operator is a tensor product \cite{Bennett1070,Donald4252}.

To characterize quantum entanglement, in \cite{Yang052423}, X. Yang et al. proposed the $q$-concurrence ($q \geq 2$) as a parameterized measure of bipartite entanglement, defined for any pure state $|\phi\rangle_{A B}$ as
\begin{eqnarray}\label{cq}
C_q\left(|\phi\rangle_{A B}\right)=1-\operatorname{tr}\left(\rho_A^q\right),
\end{eqnarray}
here $\rho_A = \operatorname{tr}_B(|\phi\rangle_{AB} \langle\phi|)$.

Note that $\mathrm{tr}(\rho_A^q)$ plays a role of $q$-purity and is a real-valued, nonnegative function for all $q$. Indeed, taking the spectral decomposition $\rho_A = {\sum_j}\, {\lambda_j}|j\rangle \langle{j}|$ in terms of the basis of eigenstates ${\{|j\rangle\}_{j = 1,\ldots, d}}$, with $0 \leq {\lambda_j} \leq 1$ and ${\sum_j}\, {\lambda_j} = 1$, it is easily verified that $\mathrm{tr}(\rho_A^q) = {\sum_j}\, {\lambda_j^{q}} \geq 0$. This measure reflects the non-purity of the reduced state and generalizes the linear entropy $ C_2 =1-\operatorname{tr} \rho^2$ ($q = 2$) \cite{Santos024101} to a one-parameter family. However, $C_q$ alone only captures the purity bias associated with $\lambda_j$, but ignores the dual quantity $(1-\lambda_j)$.

To enrich this characterization, with respect to the $q$-concurrence (\ref{cq}), we introduce its complementary dual, $\mathrm{tr}(\textmd{I}-\rho_A)-\mathrm{tr}(I-\rho_A)^q$. We define the total concurrence $C^{t}_q$ of the $q$-concurrence to be the summation of both the $q$-concurrence and its complementary dual, namely, for a quantum pure state $|\phi\rangle_{A B}$,
\begin{eqnarray}
C^{t}_q(|\phi\rangle_{A B})=1-\mathrm{tr}\rho_A^q+\mathrm{tr}(\textmd{I}-\rho_A)-  \mathrm{tr}(\textmd{I}-\rho_A)^q.\label{q1}
\end{eqnarray}
This construction symmetrically combines the original $q$-concurrence with its complementary dual, where the dual term accounts for the non-purity of the complementary operator $\textmd{I} - \rho_A$. Intuitively, this duality reflects an observer's inability to access both $\rho_A$ and its complement dual, thereby providing a more complete description of entanglement.

To illustrate the theoretical advantage of $C_q^t$, consider the difference
$$
C_q^t(|\psi\rangle_{AB}) - C_q(|\psi\rangle_{AB}) = \mathrm{tr}(\textmd{I} - \rho_A) - \mathrm{tr}(\textmd{I} - \rho_A)^q,
$$
which is always non-negative for $q \geq 2$, due to the convexity of the function $x \mapsto x^q$ on $[0,1]$ and the spectra of $\rho_A$ lying in $[0,1]$.
This term quantifies the ability of $C_q^t$ to exceed the additional entanglement captured by the original $C_q$.

Thus, the total concurrence $C_q^t$ provides a finer-grained quantization of entanglement by preserving the structure of $C_q$ and adding the complement dual. This makes it particularly useful in distinguishing entanglement in applications such as entanglement detection and resource quantization.

To illustrate the total concurrence  $C^{t}_q(\rho)$, we examine the specific case of $q = 2$.
%Consider a discrete random variable $X$ on a sample set $\{x_1, x_2,\cdots,x_n\}$. Its probability distribution is given by $p_i=\textmd{Pr}(X=x_i)$. This implies a self-information about each event $x_i$, $\textmd{I}(X=x_i)=f(p_i)$, where $f(p_i)=p_i-1$. The average information of the random variable $X$ is known as the $2$-concurrence $C_2 $ defined by
Consider a discrete random variable $Y$ defined over a sample space $\{y_1, y_2,\cdots,y_n\}$ with an associated probability distribution given by $p_i=\textmd{Pr}(Y=y_i)$. The self-information corresponding to each event $y_i$  is expressed as $\textmd{I}(Y=y_i)=f(p_i)$, where $f(p_i)=p_i-1$. The average information content of the random variable $Y$ corresponds to the $2$-concurrence $C_2 $ , which is defined as follows:
$$
C_2({\textbf{p}})=-\sum^n_{i=1}p_if(p_i),
$$
which quantifies the uncertainty before the statistical experiment.

Now, for a given event $y_i$, we define a new binomial random variable $Y_i\sim\{p_i,1-p_i\}$, which represents a complementary dual event to $y_i$, accounting for all possible outcomes except $y_i$. The corresponding self-information is given by $-f(1-p_i)$. The complementary dual of the $2-$concurrence $ C_2 $ is defined as
\begin{eqnarray}
\overline{C_2}({\textbf{p}})=-\sum^n_{i=1}(1-p_i)f(1-p_i).
\label{e2}
\end{eqnarray}

In classical information theory, uncertainty and information are intrinsically linked: the greater the uncertainty in a probability distribution, the more information is obtained upon learning the outcome of an experiment. The $2$-concurrence $ C_2 $ quantifies the information associated with the event that actually occurs, whereas the dual $2$-concurrence $C_2({\textbf{p}})$  (see Eq. (\ref{e2})) characterizes the complementary contribution. Thus, we define the total concurrence as the summation of the $2$-concurrence $ C_2 $ and its complementary dual,
\begin{eqnarray}
C^{t}_2({\textbf{p}})=\sum^n_{i=1}C_2(p_i,1-p_i)=C_2(\textbf{p})+\overline{C_2}(\textbf{p}),
\label{totalentropy}
\end{eqnarray}
where $C_2(p_i,1-p_i)$ are the $2$-concurrence  $ C_2 $ and its complementary dual of a discrete random variable $Y_i$ with a probability distribution $p_i$. This definition ensures that the total concurrence incorporates both the primary event information and its complementary counterpart, providing a more complete measure of uncertainty and information.
Therefore, similar to classical information scenarios, the present total concurrence \( C_q^t \) provides intuitively a way to include the complementary dual information. \( C_q^t \) extends the quantification of entanglement and duality beyond what is captured by the \( C_q \), providing deeper insights into the structure of entanglement.

\noindent\emph{\textbf{Example 1}}. Consider a discrete random variable $Y$ with a probability distribution $p_i=\textmd{Pr}(Y=y_i)$ corresponding to outcomes $Y=y_i$  for $i=1, 2, 3$, where $0\leq p_i\leq1$ and the normalization condition $\sum^3_{i=1}p_i=1$ holds. The $2$-concurrence  $ {C}_2(X) $  is defined as
$
C_2(Y) = -\sum_{i=1}^{3} p_i f(p_i).
$
In contrast, the total concurrence $C^{t}_2(Y)$ incorporates both the primary probability distribution and its complementary dual, given by
$
C_2^{t}(Y) = -\sum_{i=1}^{3} \left[ p_i f(p_i) + (1 - p_i) f(1 - p_i) \right].
$
As illustrated in FIG.~\ref{Fig1}, the total concurrence $C^{t}_2(Y)$ captures a higher degree of uncertainty compared to the $2$-concurrence $ C_2(Y) $, as it accounts for the contributions from both direct and complementary probability distributions.
\begin{figure}
\centering
\includegraphics[width=0.90\linewidth]{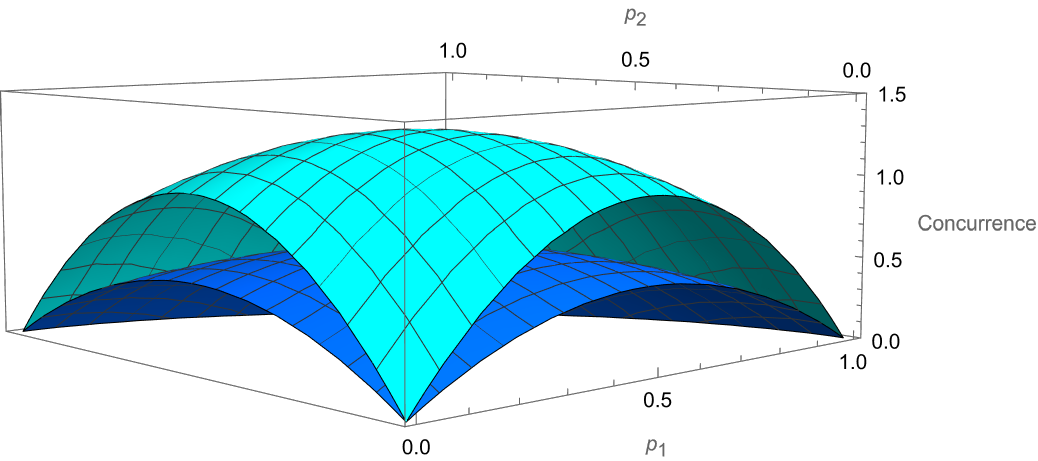}
\caption{A comparison between the $2$-concurrence $ C_2(Y) $ and the total concurrence  $C^{t}_2$  as a function of the probability distribution $\{p_i\}$ is presented. The results indicate that the total concurrence $C^{t}_2(Y)$ (light blue) consistently exceeds the $2$-concurrence $ C_2(Y) $  (blue), highlighting its ability to capture a greater degree of uncertainty by incorporating both the direct and complementary probability distributions.}
\label{Fig1}
\end{figure}

If $|\phi\rangle_{A B}$ has the Schmidt form,
\begin{equation}\label{ac1}
|\phi\rangle_{A B}=\sum_{i=1}^d\sqrt{\lambda_i}|a_ib_i\rangle,
\end{equation}
here $\sum_{i=1}^d\lambda_i=1$ and $\lambda_i\geq 0$, we have
\begin{equation}
C^{t}_q(|\phi\rangle_{A B})=d-\sum_{i=1}^d\lambda_i^{q}-\sum_{i=1}^d(1-\lambda_i)^{q},
\end{equation}
here $C^{t}_q(|\phi\rangle_{A B})$ satisfies $0 \leq C^{t}_q(|\phi\rangle_{A B}) \leq d-d^{1-q}(1+(d-1)^{q})$. The lower bound is attained if and only if $|\phi\rangle_{A B} $ is a separable state, written as $ |\phi\rangle_{A B} = \left|a_i b_i\right\rangle = \left|a_i\right\rangle \otimes \left|b_i\right\rangle $. The upper bound is reached for maximally entangled pure states, such as $\left|\Phi^{+}\right\rangle = \frac{1}{\sqrt{d}} \sum_{i=1}^d \left|a_i b_i\right\rangle $.

The parameter $q$ in the $q$-concurrence $C_q$ plays a role analogous to that in generalized entropy functions such as Tsallis and R\'{e}nyi entropies. Beyond being a formal mathematical parameter, it controls how the entanglement measure responds to the distribution of eigenvalues of the reduced density matrix $\rho_A$, thereby modulating the sensitivity to the structure of the spectra. $C_q$ quantifies the deviation from purity in a nonlinear fashion, with the nonlinearity determined by $q$. When $q \rightarrow 2$, the measure approaches the linear entropy. As $q$ increases, the measure becomes more reflective of larger eigenvalues in $\rho_A$, and thus more responsive to strongly entangled components, while becoming less sensitive to smaller eigenvalues. Concerning the total concurrence
$C_q^t(|\psi\rangle) = C_q(|\psi\rangle) + \mathrm{tr}(\textmd{I}- \rho_A) - \mathrm{tr}[(\textmd{I} - \rho_A)^q]$, the second term serves as a complementary entropy response based on the spectra of $(\textmd{I} - \rho_A)$. The parameter $q$ in this dual expression allows the measure to symmetrically enhance entanglement sensitivity and reveal the structure not captured by $C_q$ alone.

Before demonstrating that the total concurrence defined in (\ref{q1}) qualifies as a bona fide measure of quantum entanglement, we first show the following result.

\begin{lemma}
The function $F^{t}_q(\rho)=1-\mathrm{tr}\rho^q+\mathrm{tr}(\textmd{I}-\rho)-  \mathrm{tr}(I-\rho)^q$ is concave, meaning that for any probability distribution $\left\{p_i\right\}$ and corresponding density matrices $\rho_i$, the following inequality holds
$$\sum^n_{i=1}p_iF^{t}_q(\rho_i)\leq F^{t}_q(\rho),$$
for any $q\geq2$. Equality occurs if and only if all $\rho_i$ are identical for all $p_i>0$.
\label{LemmaA}
\end{lemma}

\emph{Proof.}
For a given quantum state $\rho$, a generalized entropy function can be defined as follows, as introduced in \cite{Canosa170401}
\begin{eqnarray}
S_f(\rho)=\mathrm{tr}[f(\rho)]=\sum_if(\lambda_i),
\label{general0}
\end{eqnarray}
where $\lambda_i$ are the eigenvalues of $\rho$, and $f$  is a smooth concave function satisfying $f''(\lambda) < 0$  for $\lambda \in (0, 1)$. Additionally, $f(\lambda)$  is continuous on $[0, 1]$ and fulfills the boundary conditions $f(0)=f(1)=0$. Consider a function $ h: [0,1] \to \mathbb{R} $ that maps the interval $[0,1]$ to the set of real numbers,
\begin{eqnarray}
h(x)=x-x^q+1-x-(1-x)^q=1-x^q-(1-x)^q.
\label{q2}
\end{eqnarray}
The function $h(x)$ is a smooth and strictly concave real-valued function defined on the interval $x\in [0,1]$, with boundary conditions $h(0)=h(1)=0$. Its second derivative is given by $h''(x) = -x(x-1) \left[(1-x)^{q-2} + x^{q-2} \right] < 0, \quad \text{for any } q \geq 2.$ This confirms that $h(x)$ exhibits concavity over the specified domain.

Therefore, for the entropy form (\ref{general0}) corresponding to $F^{t}_q(\rho)$, we know
$$F^{t}_q(\rho)={\rm tr}[h(\rho)]=\sum_{i}h(\lambda_i).$$

%In fact, for any concave function $h(x)$, the functional $\rm{tr}[h(\rho)]$ is concave for any Hermitian operator $\rho$ (see Sec. III in Ref \cite{Renou070403} for details). The concavity of $h$ ensures the concavity of $F^{t}_q(\rho)$, which proves the concavity. So for any convex combination $\rho=\sum^n_ip_i\rho_i$, we have $\sum^n_ip_i F^{t}_q(\rho_i)\leq  F^{t}_q(\rho).$ The equality holds if and only if all $\rho_i$'s are equal.
%In fact, for any concave function $h(x)$, the functional $\rm{tr}[h(\rho)]$  remains concave for any Hermitian operator $\rho$ (see Sec. III in Ref. \cite{Renou070403} for further details). The concavity of $ h(x)$ directly implies the concavity of $F^{t}_q(\rho)$, thereby establishing its concavity property. Consequently, for any convex combination $\rho=\sum^n_ip_i\rho_i$, the following inequality holds $\sum_{i}^{n} p_i F^{t}_q(\rho_i) \leq F^{t}_q(\rho).$ Equality is attained if and only if all \(\rho_i\) are identical.$\Box$

For any concave function $h(x)$, it follows that the functional $\rm{tr}[h(\rho)]$ remains concave for any Hermitian operator $\rho$ (for further details, see Sec. III in Ref. \cite{Renou070403}). The concavity of $ h(x)$ directly leads to the concavity of $F^{t}_q(\rho)$, thereby confirming its concave nature. Consequently, for any convex combination $\rho=\sum^n_ip_i\rho_i$, the inequality $\sum_{i}^{n} p_i F^{t}_q(\rho_i) \leq F^{t}_q(\rho)$ holds. Equality is achieved if and only if all the density matrices \(\rho_i\)  are identical. $\Box$

%According to the above analysis, based on the total concurrence ${C}^t_q$ defined in (\ref{q1}), we propose a new normalized bipartite entanglement measure, called $\mathcal{C}^{t}_q$-concurrence measure.

The Lemma 1 primarily establishes the concavity of the function $F^{t}_q$, which is to demonstrate that the total concurrence satisfies the axiomatic requirements of (E3): entanglement measure does not increase on average under LOCC, within the entanglement quantification framework. Building on the above analysis and leveraging the definition of total concurrence  ${C}^t_q$ in (\ref{q1}), we introduce a new normalized bipartite entanglement measure, termed the $\mathcal{C}^{t}_q$-concurrence measure.

For a pure bipartite state $|\Phi\rangle_{AB}$ in the Hilbert space ${\cal H}_A\otimes {\cal H}_{B}$, we define its $\mathcal{C}^{t}_q$-concurrence as
\begin{eqnarray}
\mathcal{C}^{t}_q(|\Phi\rangle_{AB})=\frac{1}{\mu}{C}^t_q(|\Phi\rangle_{AB})
=\frac{1}{\mu}{F}^t_q(\rho_A),
\label{Cq1}
\end{eqnarray}
here $\rho_A={\rm tr}_B(|\Phi\rangle_{AB}\langle\Phi|)$, and $\mu=d-d^{1-q}(1+(d-1)^{q})$  serves as the normalization factor to ensure the proper scaling of the measure.

For a bipartite mixed state $\rho_{AB}$, the $\mathcal{C}^{t}_q$-concurrence is defined through the  convex-roof extension, which is expressed as
\begin{eqnarray}
\mathcal{C}^{t}_q(\rho_{AB})=\inf_{\{p_i,|\Phi_i\rangle\}}\sum_ip_i\mathcal{C}^{t}_q(|\Phi_i\rangle_{AB}),
\label{Cq2}
\end{eqnarray}
where the infimum is taken over all the possible pure state decompositions of $\rho_{AB}=\sum_ip_i|\Phi_i\rangle_{AB}\langle\Phi_i|$, with probabilities $ p_i \geq 0 $ satisfying the normalization condition $ \sum_i p_i = 1$.

%(E1) From the non-negativity in Lemma \ref{LemmaA} and  (\ref{Cq2}), it follows that $\mathcal{C}_q^t(\rho_{AB})\geq 0$ for any bipartite state $\rho_{AB}$, where the equality holds if and only if $\rho_{AB}$  is separable.
\begin{theorem}
The $\mathcal{C}^{t}_q$-concurrence $(q\geq 2)$ defined in (\ref{Cq2}) constitutes a well-established  parameterized measure of bipartite entanglement.
\label{theoremA}
\end{theorem}

\emph{Proof.}
(E1) From the definition in (\ref{q1}) and (\ref{Cq1}), we  get $\mathcal{C}^{t}_q(|\Phi\rangle_{AB})\geq 0$. By the convex-roof extension of $\mathcal{C}^{t}_q(\rho_{AB})$, we have $\mathcal{C}^{t}_q(\rho_{AB})\geq 0$.

Next, we prove that $\mathcal{C}^{t}_q(\sigma_{AB})=0$ if and only if $\sigma_{AB}$ is separable. When a quantum state $|\Phi\rangle_{AB}$ is pure separable, the reduced state $\rho_{A}$ is also pure. Therefore, $\rho_{A}^q = \rho_{A}$, $\mathrm{tr}{\rho_{A}}=\mathrm{tr}\rho_{A}^q=1$ and $\mathrm{tr}(I-\rho_{A})=\mathrm{tr}(I-\rho_{A})^q=d-1$. Thus, $\mathcal{C}^{t}_q(|\Phi\rangle_{AB})=\frac{1}{\mu}{F}^t_q(\rho_A)
=\frac{1}{\mu}(1-\mathrm{tr}\rho_A^q+\mathrm{tr}(\textmd{I}-\rho_A)-  \mathrm{tr}(I-\rho_A)^q)=0$. Since any bipartite separable mixed states can be written in the form of convex combinations of bipartite separable pure states, that is, $\sigma_{AB}=\sum_ip_i|\Psi_i\rangle_{AB}\langle\Psi_i|$, where $|\Psi_i\rangle_{AB}$ are bipartite separable pure states. Thus, be definition $\mathcal{C}^{t}_q(\sigma_{AB})=0$. Conversely, when the quantum state $|\Phi\rangle_{AB}$ is a pure state, one has
$\mathcal{C}^{t}_q(|\Phi\rangle_{AB})=\frac{1}{\mu}{F}^t_q(\rho_A)$.
let $\lambda_i$ $(i=1\cdots d)$ be the corresponding eigenvalue of $\rho_A$. Then ${F}^t_q(\rho_A)=d-\sum_{i=1}^d\lambda_i^{q}-\sum_{i=1}^d(1-\lambda_i)^{q}=0$, namely, $d-\sum_{i=1}^d[\lambda_i^{q}+(1-\lambda_i)^{q}]=0$, where $0\leq \lambda_i\leq1$ and $q\geq 2$, which implies that either $\lambda_i=1$ or $\lambda_i=0$. Hence, $\rho_A=|\Phi\rangle_{A}\langle\Phi|$ is a pure state and $|\Phi\rangle_{AB}=|\Phi\rangle_{A}\otimes|\Phi\rangle_{B}$ is separable state.

%Firstly,  when the quantum state is a pure  state $|\Phi\rangle_{AB}$,$\mathcal{C}^{t}_q(|\Phi\rangle_{AB})=\frac{1}{\mu}{F}^t_q(\rho_A)=$ , the eigenvalues of their density matrices have the special property that one of them is 1 and the rest are 0, so $\rho_{A}^q = \rho_{A}$, $\mathrm{tr}{\rho_{A}}=\mathrm{tr}\rho_{A}^q=1, $ and $\mathrm{tr}(I-\rho_{A})=\mathrm{tr}(I-\rho_{A})^q=d-1$. Thus, $\mathcal{C}^{t}_q(|\Phi\rangle_{AB})=\frac{1}{\mu}{F}^t_q(\rho_A)=\frac{1}{\mu}(1-\mathrm{tr}\rho_A^q+\mathrm{tr}(\textmd{I}-\rho_A)-  \mathrm{tr}(I-\rho_A)^q)=0$. Since bipartite separable mixed states can be written in the form of convex combinations of bipartite separable pure states, that is $\sigma_{AB}=\sum_ip_i|\Psi_i\rangle_{AB}\langle\Psi_i|$, where $|\Psi_i\rangle_{AB}$ is bipartite separable pure state. Thus by the above analysis and (\ref{Cq2}),  it can be shown that $\mathcal{C}^{t}_q(\sigma_{AB})=0$.

If $\rho_{AB}$ is an entangled state, then in any pure state decomposition of $\rho_{AB}$, there exists at least one entangled pure state $|\Phi_i\rangle_{AB}$ such that $\mathcal{C}^{t}_q(|\Phi_i\rangle_{AB})> 0$. Hence, $\mathcal{C}^{t}_q(\rho_{AB})> 0$.

(E3) We first show that $\mathcal{C}^{t}_q$-concurrence is not increasing on average under  LOCC for any pure state $|\Phi\rangle_{A B}$. Let $\left\{M_i\right\}$ be a CPTP map acting on the subsystem $B$. The postmap states are
$$
\rho_i^{A B}=\frac{1}{p_i} M_i\left(|\Phi\rangle_{A B}\langle\Phi|\right)
$$
with probability $p_i=\operatorname{tr}\left[M_i\left(|\Phi\rangle_{A B}\langle\Phi|\right)\right]$. We have
$$
\rho^A=\sum_i p_i \rho_i^A,
$$
where $\rho^A=$ $\operatorname{tr}_B\left(|\Phi\rangle_{A B}\langle\Phi|\right)$ and $\rho_i^A=\operatorname{tr}_B \rho_i^{A B}$.

Let $\left\{q_{i l},\left|\Phi_{i l}\right\rangle_{A B}\right\}$ represent the optimal pure-state decomposition of $\rho_i^{A B}$, satisfying the condition that
$$
\mathcal{C}^{t}_q\left(\rho_i^{A B}\right)=\sum_l q_{i l} \mathcal{C}^{t}_q\left(\left|\Phi_{i l}\right\rangle_{A B}\right),
$$
where $q_{i l} \geq 0, \quad \sum_l q_{i l}=1$ and $\rho_i^{A B}=\sum_l q_{i l}\left|\Phi_{i l}\right\rangle_{A B}\left\langle\Phi_{i l}\right|$. Then, we obtain the following relation
\begin{eqnarray}
\mathcal{C}^{t}_q\left(|\Phi\rangle_{A B}\right)&=&\frac{1}{\mu}{C}^t_q(\rho^A)
=\frac{1}{\mu}{C}^t_q\left(\sum_i p_i \rho_i^A\right)\nonumber\\
&=&\frac{1}{\mu}{C}^t_q\left(\sum_{i l} p_i q_{i l} \rho_{i l}^A\right)
\geq \sum_{i l} p_i q_{i l} \frac{1}{\mu}{C}^t_q\left(\rho_{i l}^A\right)\nonumber\\
&=&\sum_{i l} p_i q_{i l} \mathcal{C}^{t}_q\left(\left|\Phi_{i l}\right\rangle_{A B}\right)
=\sum_i p_i \mathcal{C}^{t}_q\left(\rho_i^{A B}\right),
\label{slocc1}
\end{eqnarray}
where $\rho_{i l}^A=\operatorname{tr}_B\left(\left|\Phi_{i l}\right\rangle_{A B}\left\langle\Phi_{i l}\right|\right)$ and the inequality is from the concavity shown in  Lemma \ref{LemmaA},
this concludes the proof for the case of pure states.

For mixed states \(\rho_{AB}\), we define $\rho_i^{A B}=\frac{1}{p_i} M_i\left(\rho_{A B}\right)$, where $p_i=\operatorname{tr}\left[M_i\left(\rho_{A B}\right))\right]$. Let $\left\{t_j,\left|\upsilon_j\right\rangle_{A B}\right\}$ represent the optimal pure state ensemble of $\rho_{A B}$. For each $j$, we define $\rho_{j i}^{A B}=\frac{1}{t_{j i}} M_i\left(\left|\upsilon_j\right\rangle_{A B}\right)$, which is the final state obtained after the action of $M_i$ on $\left|\upsilon_j\right\rangle_{A B}$, where $t_{ji}=\operatorname{tr}\left[M_i\left(\left|\upsilon_j\right\rangle_{A B}\right)\right]$.
(\ref{slocc1}) implies that
\begin{eqnarray}
\mathcal{C}^{t}_q\left(\left|\upsilon_j\right\rangle_{A B}\right) \geq \sum_i t_{j i} \mathcal{C}^{t}_q\left(\rho_{j i}^{A B}\right).
\label{slocc2}
\end{eqnarray}
By the linearity of $M_i$, we have
\begin{eqnarray}
\rho_i^{A B} & =&\frac{1}{p_i} M_i\left(\rho_{A B}\right)
=\frac{1}{p_i} M_i\left(\sum_j t_j\left|\upsilon_j\right\rangle_{A B}\right) \nonumber\\
& =&\frac{1}{p_i} \sum_j t_j M_i\left(\left|\upsilon_j\right\rangle_{A B}\right)
=\frac{1}{p_i} \sum_j t_j t_{j i} \rho_{j i}^{A B}.
\label{slocc3}
\end{eqnarray}
For every index pair $(j, i)$, let $\left\{t_{j i l},\left|\Phi_{j i l}\right\rangle_{A B}\right\}$  represent the optimal pure-state decomposition of $\rho_{j i}^{A B}$.
Then, from  (\ref{slocc2}) we obtain
$$
\begin{aligned}
\mathcal{C}^{t}_q\left(\rho_{A B}\right) & =\sum_j t_j
\mathcal{C}^{t}_q\left(\left|\upsilon_j\right\rangle_{A B}\right) \\
& \geq \sum_{j i} t_j t_{j i} \mathcal{C}^{t}_q\left(\rho_{j i}^{A B}\right) \\
& =\sum_{j i l} t_j t_{j i} t_{j i l}\mathcal{C}^{t}_q\left(\left|\Phi_{j i l}\right\rangle_{A B}\right) \\
& =\sum_i p_i \mathcal{C}^{t}_q\left(\rho_i^{A B}\right),
\end{aligned}
$$
where in the final equality, we have utilized the result from (\ref{slocc3}),  that $\mathcal{C}^{t}_q\left(\rho_i^{AB}\right) = \sum_{jl} \frac{t_j t_j t_{ji l}}{p_i} \mathcal{C}^{t}_q\left(|\Phi_{ji l}\rangle_{AB}\right)$. Additionally, the mixed state $\rho_i^{AB}$ is expressed as $\rho_i^{AB} = \frac{1}{p_i} \sum_j t_j t_{ji} t_{ji l} |\Phi_{ji l}\rangle_{AB} \langle\Phi_{ji l}|$.

(E4) The convexity (E4) is followed directly from the fact that all the measures constructed via the convex roof extension are convex \cite{Horodecki865}.

%We adopt the approach given in \cite{Mintert207} to show that our entanglement measure does not increase under LOCC. Denote $\vec{\lambda}_\psi$ $(\vec{\lambda}_\phi$ the Schmidt vector given by the squared Schmidt coefficients of the state $|\psi\rangle$ $(|\phi\rangle)$ in the decreasing order.It has been shown that the state $|\phi\rangle$ can be prepared starting from the state $|\psi\rangle$ under LOCC if and only if $\vec{\lambda}_\psi$ is majorized by $\vec{\lambda}_\phi$ \cite{Nielsen436}, $\vec{\lambda}_\psi \prec \vec{\lambda}_\phi$, where the majorization means that the components $\left[\lambda_\psi\right]_i$ $\left[\lambda_\phi\right]_i$ of $\vec{\lambda}_\psi$ $(\vec{\lambda}_\phi$, listed in nonincreasing order, satisfy $\sum_{i=1}^j\left[\lambda_\psi\right]_i \leqslant \sum_{i=1}^j\left[\lambda_\phi\right]_i$ for $1<j \leqslant d$.Since the entanglement cannot increase under LOCC, any entanglement measure $E$ has to satisfy that $E(\psi) \geq E(\phi)$ whenever $\vec{\lambda}_\psi \prec \vec{\lambda}_\phi$. This condition, known as the Schur concavity, is satisfied if and only if $E$, given as a function of the squared Schmidt coefficients $\lambda_i$ 's \cite{Ando163}, is invariant under the permutations of any two arguments and satisfies$$\left(\lambda_i-\lambda_j\right)\left(\frac{\partial E}{\partial \lambda_i}-\frac{\partial E}{\partial \lambda_j}\right) \leq 0$$for any two components $\lambda_i$ and $\lambda_j$ of $\vec{\lambda}$.

(E5) We follow the method outlined in \cite{Mintert207} to demonstrate that our entanglement measure remains non-increasing under LOCC. Let $\vec{\lambda}_\alpha$ and $(\vec{\lambda}_\beta$ represent the Schmidt vectors corresponding to the states $|\alpha\rangle$ and  $(|\beta\rangle)$, respectively, where each vector is ordered in decreasing order of the squared Schmidt coefficients. It has been established that the state  $|\beta\rangle$ can be obtained from the state $|\alpha\rangle$ through  LOCC if and only if the Schmidt vector $\vec{\lambda}_\alpha$  is majorized by $\vec{\lambda}_\beta$ \cite{Nielsen436}, denoted as $\vec{\lambda}_\alpha \prec \vec{\lambda}_\beta$.  Majorization implies that the components $\left[\lambda_\alpha\right]_i$  and $\left[\lambda_\beta\right]_i$ of $\vec{\lambda}_\alpha$  and $(\vec{\lambda}_\beta$, respectively, listed in non-increasing order, satisfy the condition $\sum_{i=1}^j\left[\lambda_\alpha\right]_i \leq \sum_{i=1}^j\left[\lambda_\beta\right]_i$ for $1<j \leq d$.

As entanglement is non-increasing under LOCC, any  entanglement measure $\mathcal{E}$ must satisfy the condition $\mathcal{E}(\alpha) \geq \mathcal{E}(\beta)$ whenever $\vec{\lambda}_\alpha \prec \vec{\lambda}_\beta$. This requirement, referred to as Schur concavity, holds if and only if $\mathcal{E}$, when formulated as a function of the squared Schmidt coefficients $\lambda_i$, remains unchanged under the exchange of any two coefficients.
Additionally, it must satisfy the inequality
$$
\left(\lambda_i-\lambda_j\right)\left(\frac{\partial \mathcal{E}}{\partial \lambda_i}-\frac{\partial \mathcal{E}}{\partial \lambda_j}\right) \leq 0,
$$
for any pair of components $\lambda_i$ and $\lambda_j$ of the Schmidt vector $\vec{\lambda}$.

For any pure state \( |\psi\rangle \) represented in the form given by (\ref{ac1}),
the $\mathcal{C}^{t}_q$-concurrence is clearly invariant under permutations of the Schmidt coefficients. Let us define the function $f(\lambda)=(1-\lambda)^{q-1}-{\lambda}^{q-1}$. By differentiating $f(\lambda)$, we obtain $f'(\lambda)=-(q-1)(1-\lambda)^{q-2}-{\lambda}^{q-2}<0$, this result indicates that $f(\lambda)$ is a strictly decreasing function. Therefore, $f(\lambda)$  is a subtractive function, meaning that it produces a negative contribution as $ \lambda $ increases.

Since the following holds
$\left(\lambda_i-\lambda_j\right)\left(\frac{\partial \mathcal{C}^{t}_q}{\partial \lambda_i}-\frac{\partial \mathcal{C}^{t}_q}{\partial \lambda_j}\right)=$
$q\left(\lambda_i-\lambda_j\right)\left((1-\lambda_i)^{q-1}
-(1-\lambda_j)^{q-1}-{\lambda_i}^{q-1}+{\lambda_j}^{q-1}\right) \leq 0,$
for any two components $\lambda_i$ and $\lambda_j$ of the squared Schmidt coefficients of the state $|\psi\rangle$, we observe the following, if $\lambda_{i}<\lambda_{j}$, then $f(\lambda_i)>f(\lambda_j)$, implying $\left(\lambda_i-\lambda_j\right)\left(\frac{\partial \mathcal{C}^{t}_q}{\partial \lambda_i}-\frac{\partial \mathcal{C}^{t}_q}{\partial \lambda_j}\right)\leq 0$. Conversely, if $\lambda_{i}>\lambda_{j}$, then $f(\lambda_i) < f(\lambda_j)$, yielding $\left(\lambda_i-\lambda_j\right)\left(\frac{\partial \mathcal{C}^{t}_q}{\partial \lambda_i}-\frac{\partial \mathcal{C}^{t}_q}{\partial \lambda_j}\right)\leq 0$.
Therefore, this condition ensures that the $ \mathcal{C}^{t}_q $-concurrence is monotonic under LOCC, meaning
$\mathcal{C}^{t}_q(|\psi\rangle) \geq \mathcal{C}^{t}_q(\Lambda |\psi\rangle),$
for any LOCC operation  $\Lambda $. This proves that $ \mathcal{C}^{t}_q $ is an entanglement measure that does not increase under LOCC.

Next let $\rho=\sum_i p_i\left|\psi_i\right\rangle\left\langle\psi_i\right|$ be the optimal pure state decomposition of $\mathcal{C}^{t}_q(\rho)$ with $\sum_i p_i=1$ and $p_i>0$. We can get
$$
\begin{aligned}
\mathcal{C}^{t}_q(\rho) & =\sum_i p_i \mathcal{C}^{t}_q\left(\left|\psi_i\right\rangle\right) \\
& \geq \sum_i p_i \mathcal{C}^{t}_q\left(\Lambda\left|\psi_i\right\rangle\right) \\
& \geq \mathcal{C}^{t}_q(\Lambda(\rho)),
\end{aligned}
$$
where the last inequality is from the definition (\ref{Cq2}). $\Box$

%By (E5) we can easily prove that $\mathcal{C}^{t}_q(\rho)$ satisfies condition (E2).(E2) Since $\mathcal{C}^{t}_q(\rho)\geq \mathcal{C}^{t}_q(U_A \otimes U_B \rho U_A^{\dagger}\otimes U_B^{\dagger} ) =\mathcal{C}^{t}_q\left(\rho^{\prime}\right) \geq \mathcal{C}^{t}_q\left(U_A^{\dagger} \otimes U_B^{\dagger} \rho^{\prime} U_A \otimes U_B\right)= \mathcal{C}^{t}_q(\rho)$  with $\rho^{\prime}=U_A \otimes U_B \rho U_A^{\dagger} \otimes U_B^{\dagger}$ from condition (E5), $U_A$ and $U_B$ are local unitaries on  ${\cal H}_A$ and ${\cal H}_{B}$, respectively, we prove that $\mathcal{C}^{t}_q(\rho)$ is invariant under local unitary transformations. This implies the condition (E2).

%In \cite{Yang052423} a parameterized entanglement monotone called $q$-concurrence $C_q\left(\ket{\psi}\right)$ has been introduced for any pure state $\ket{\psi}$, $C_q\left(\ket{\psi}\right)=1-\mathrm{tr}\rho_A^q$ for $q\geq 2$. In \cite{WEI275303} a new parameterized entanglement measure, called $\alpha$-concurrence, was proposed,
%$C_{\alpha}\left(\ket{\psi}_{AB}\right)=\mathrm{Tr}\rho_A^{\alpha}-1$for $0\leq \alpha\leq1/2$. It seems that $\alpha$-concurrence is in some sense dual to the $q$-concurrence as $\alpha\in [0,1/2]$ and $q\geq 2$. Similarly, with respect to our $\mathcal{C}^{t}_q$, we may introduce a $\alpha$-concurrence complementary dual information and propose the $\mathcal{C}^{t}_\alpha$-concurrence $(0\leq \alpha \leq \frac{1}{2})$, see Appendix \ref{a}.

In \cite{Yang052423}, a parameterized entanglement monotone known as the $q$-concurrence, $C_q\left(\ket{\psi}\right)$, was introduced for any pure state $\ket{\psi}$, defined as $C_q\left(\ket{\psi}\right)=1-\mathrm{tr}\rho_A^q$  for $q\geq 2$. In \cite{WEI275303}, a new parameterized entanglement measure, referred to as the $\alpha$-concurrence, was proposed, which is given by $C_{\alpha}\left(\ket{\psi}_{AB}\right)=\mathrm{tr}\rho_A^{\alpha}-1$ for $0\leq \alpha\leq1/2$.
The $\alpha$-concurrence can be regarded as a complementary counterpart to the $q$-concurrence, with the parameters constrained by $0\leq \alpha \leq \frac{1}{2}$ and $q\geq 2$.
Following a similar reasoning, with respect to our $ \mathcal{C}^{t}_q $-concurrence, we may introduce a complementary dual information measure, termed the $ \alpha $-concurrence. We then propose the $ \mathcal{C}^{t}_\alpha $-concurrence for $ 0 \leq \alpha \leq \frac{1}{2} $, as detailed in Appendix \ref{a}.

\section{establishing bounds for the $\mathcal{C}^{t}_q$-concurrence measure}\label{sect3}
%Owing to the optimization in estimating entanglement measures, it is generally difficult to obtain analytical expressions of the entanglement measures for general mixed states. In this section, we derive analytical lower bounds for the $\mathcal{C}^{t}_q$-concurrence based on PPT and realignment criteria \cite{PhysRevA.59.4206,rudolph2005further}.
Due to the complexities involved in optimizing entanglement measures, it is typically challenging to derive explicit analytical expressions for such measures, especially in the case of general mixed states.
In this section, we derive analytical lower bounds for the $\mathcal{C}^{t}_q$-concurrence using the Positive Partial Transpose (PPT) criterion and the realignment criterion \cite{Horodecki4206, Rudolph219}.

 %In this section, we present analytical lower bounds for the $\mathcal{C}^{t}_q$-concurrence. These bounds are obtained by leveraging the Positive Partial Transpose (PPT) criterion and the realignment criterion \cite{PhysRevA.59.4206, rudolph2005further}.

%A bipartite state can be written as $\rho=\sum_{ijkl}\rho_{ij,kl}|ij\rangle\langle kl|$, where the subscripts $i$ and $k$ ($j$ and $l$) are the row and column indices for the subsystem $A$ ($B$), respectively. The PPT criterion says that if the state $\rho$ is separable, then the partial transposed matrix $\rho^{\Gamma}=\sum_{ijkl}\rho_{ij,kl}|il\rangle\langle kj|$ with respect to the subsystem $B$ is non-negative, $\rho^{\Gamma}\geq 0$. While the realignment criterion says that if $\rho$ is separable, the realigned matrix of $\rho$, $\mathcal{R}\left(\rho\right)=\sum_{ijkl}\rho_{ij,kl}|ik\rangle\langle jl|$, satisfies that $\|\mathcal{R}\left(\rho\right)\|_1\leq 1$, where $\|X\|_1=\mathrm{tr}\sqrt{X^{\dagger}X}$ denotes the trace norm of matrix $X$.

A general bipartite state can be expressed as $\rho=\sum_{ijkl}\rho_{ij,kl}|ij\rangle\langle kl|$, where the indices $i$ and $k$ ($j$ and $l$) correspond to the row and column indices for subsystems $A$ (and $B$), respectively. According to the PPT  criterion, for a state $\rho$ to be separable, its partial transpose with respect to subsystem $B$ , denoted $\rho^{\Gamma}=\sum_{ijkl}\rho_{ij,kl}|il\rangle\langle kj|$, must be non-negative, i.e., $\rho^{\Gamma}\geq 0$.
On the other hand, the realignment criterion states that if $\rho$ is separable, the realigned matrix $\mathcal{R}\left(\rho\right)=\sum_{ijkl}\rho_{ij,kl}|ik\rangle\langle jl|$, must satisfy $\|\mathcal{R}\left(\rho\right)\|_1\leq 1$, where $\|M\|_1=\mathrm{tr}\sqrt{M^{\dagger}M}$ represents the trace norm of the matrix $M$.

For $\rho=|\Phi\rangle\langle\Phi|$, where $\ket{\Phi}$ is the pure state defined in  (\ref{ac1}), it was shown in \cite{Chen040504} that
\begin{equation}\label{text1}
1\leq\left\|\rho^{\Gamma}\right\|_1
=\left\| \mathcal{R}(\rho)\right\|_1
=\left(\sum_{i=1}^r\sqrt{\lambda_i}\right)^2\leq r.
\end{equation}

Specifically, when $q=2$, the $q$-concurrence is given by $C_2\left(\ket{\Phi}\right)=1-\sum_{i=1}^d\lambda_i^2=2\sum_{i<j}\lambda_i\lambda_j$. In \cite{Chen040504}, an analytical lower bound for $C_2\left(\ket{\Phi}\right)$ was derived
\begin{equation}\label{b2p}
C_2\left(\ket{\Phi}\right)\geq\frac{1}{d\left(d-1\right)}
\left(\left\|\rho^{\Gamma}\right\|_1-1\right)^2.
\end{equation}
This bound holds for any pure state $\ket{\Phi}$ in the Hilbert space $\mathcal{H}_A\otimes\mathcal{H}_B$. Additionally, for $q=2$, the $\mathcal{C}^{t}_q$-concurrence  is expressed as $\mathcal{C}^{t}_2\left(\ket{\Phi}\right)=2-2\sum_{i=1}^d\lambda_i^2
=4\sum_{i<j}\lambda_i\lambda_j=2C_2\left(\ket{\Phi}\right)$.

Before establishing a tight lower bound for the $\mathcal{C}^{t}_q$-concurrence in the regime $q\geq2$, we first derive the following result. Given a pure state $|\Phi\rangle$ as defined in (\ref{ac1}), we examine the monotonicity of the function $f\left(q\right)$ , which is defined as
\begin{align}\label{bound1}
f\left(q\right)=\frac{d-\sum_{l=1}^d\lambda_l^{q}-\sum_{l=1}^d(1-\lambda_l)^{q}}
{d-d^{1-q}(1+(d-1)^{q})},
\end{align}
for any $q\geq2$.
We find the  expression the $\frac{\partial f}{\partial q} $
\begin{equation}\label{text006}
\frac{\partial f}{\partial q}=\frac{d^{q-1}}{\left(-d^q+(d-1)^q+1\right)^2} M_{dq},
\end{equation}
where
%\begin{widetext}
%\begin{align}\label{bound2}
%M_{dq}&=-\left(d^q-(d-1)^q-1\right)\left(\sum _{i=1}^d \ln \left(1-\lambda _l\right) \left(1-\lambda _l\right){}^q+\sum _{i=1}^d \ln \left(\lambda _l\right) \lambda _l^q\right) \nonumber\\
%&-d^{1-q} \left((d-1)^q (\ln d-\ln (d-1))+\ln d\right) \left(-\sum _{i=1}^d \left(1-\lambda _l\right){}^q-\sum _{i=1}^d \lambda _l^q+d\right).
%\end{align}
%\end{widetext}
\begin{align}\label{bound2}
&M_{dq}=-\left(d^q-(d-1)^q-1\right)\nonumber\\
&\left(\sum _{i=1}^d \ln \left(1-\lambda _l\right) \left(1-\lambda _l\right){}^q
+\sum _{i=1}^d \ln \left(\lambda _l\right) \lambda _l^q\right) \nonumber\\
&-d^{1-q} \left((d-1)^q (\ln d-\ln (d-1))+\ln d\right)  \nonumber\\
&\left(-\sum _{i=1}^d \left(1-\lambda _l\right){}^q-\sum _{i=1}^d \lambda _l^q+d\right).
\end{align}

%We see that the sign of the first derivative of $f\left(q\right)$ with respect to $q$ depends on the sign of the function $M_{dq}$ under constraints $\sum_{i=1}^d\lambda_i=1$ and  $\lambda_i>0$ for $i=1,...,d$. Consider the minimum of $M_{dq}$ by using Lagrange multipliers \cite{PhysRevA.67.012307} subject to the constraints $\sum_i\lambda_i=1$ and $\lambda_i>0$. There is only one stable point under the constraints, $\lambda_i=1/d$, $i=1,\cdots,d$, for which we have $M_{dq}=0$.

We observe that the sign of  $\frac{\partial f}{\partial q} $ is determined by the sign of the function $M_{dq}$, subject to the conditions $\sum_{l=1}^d\lambda_l=1$ and  $\lambda_l>0$ for $l=1,...,d$. To investigate the minimum of $M_{dq}$, we apply the method of Lagrange multipliers \cite{Rungta012307} while imposing the constraints $\sum_l\lambda_l=1$ and $\lambda_l>0$. Under these conditions, there exists a single stable point at $\lambda_l=1/d$ for all $l=1,\cdots,d$, where we find that $M_{dq}=0$.

%The second derivative of $f\left(q\right)$ with respect to $q$ at the stable point is given by\begin{align}\label{bound3}\frac{\partial^2M_{dq}}{\partial\lambda_i^2}\Big|_{\lambda_i=\frac{1}{d}}=& d^{1-q} \left(\left(-d^q+(d-1)^q+1\right) \left((d-1)^q \ln \left(\frac{d-1}{d}\right)+\ln \left(\frac{1}{d}\right)\right)\right)\nonumber\\+&d \left(\left(\frac{1}{d}\right)^q+\left(\frac{d-1}{d}\right)^q-1\right) \left((d-1)^q (\ln d-\ln (d-1))+\ln d\right),\end{align}which is semipositive when $q\geq s\equiv3.33902$ for $d=2$ and $q\geq2$ for $d\geq3$. In these cases, the minimum extreme point is the minimum point and $\partial f/\partial q\geq 0$.From the above analysis, it is straightforward to prove the following conclusion.

The $\frac{\partial^2}{\partial q} f\left(q\right)$ at the stable point is given by
\begin{widetext}
\begin{align}
\label{bound3}
\frac{\partial^2M_{dq}}{\partial\lambda_i^2}\Big|_{\lambda_i=\frac{1}{d}}=& d^{1-q} \left(\left(-d^q+(d-1)^q+1\right) \left((d-1)^q \ln \left(\frac{d-1}{d}\right)+\ln \left(\frac{1}{d}\right)\right)\right)\nonumber\\+&d \left(\left(\frac{1}{d}\right)^q+\left(\frac{d-1}{d}\right)^q-1\right) \left((d-1)^q (\ln d-\ln (d-1))+\ln d\right),
\end{align}
\end{widetext}
which is non-negative when $q\geq s\equiv3.33902$ for $d=2$ and for all $q\geq2$ when $d\geq3$. In these parameter ranges, the extremum corresponds to a minimum, ensuring that $\partial f/\partial q\geq 0$.
Following this analysis, we arrive at the following conclusion.
\begin{corollary}\label{corollary}
For  any $q\geq h$, the following bound holds
$$
\mathcal{C}^{t}_q(\rho)\geq\frac{d-d^{1-q}(1+(d-1)^{q})}{d-d^{1-h}(1+(d-1)^{h})}
\mathcal{C}^{t}_h(\rho),
$$
where the parameter $h$ satisfies $h\geq s=3.33902$  for  $d=2$ and $h \geq 2$ when $d\geq3$.
\end{corollary}
The Corollary \ref{corollary}proposes a lower bound on the monotonicity of the total concurrence $\mathcal{C}^{t}_q(\rho)$ for any $q$. Specifically, for any $q \geq h$, $\mathcal{C}^{t}_q(\rho)$ can be estimated as a function of $q$ and $h$. This means that, under appropriate conditions, the total concurrence does not decrease as $q$ increases. This monotonicity allows one to estimate the entanglement for different values of $q$. Physically, Corollary \ref{corollary} suggests that the sensitivity of the entanglement increases as $q$ increases, thus providing a range of adjustments to the parameters such that the actual estimation does not require a complete recalculation for each $q$.
By utilizing Corollary \ref{corollary}, we give the main conclusions of our analysis as follows.
\begin{theorem}\label{thma}
For a general bipartite state $\rho\in\mathcal{D}$, we obtain the following results
\small{\begin{equation}
\label{text100}
\mathcal{C}^{t}_q\left(\rho\right)\geq\frac{d-d^{1-q}(1+(d-1)^{q})}{(d-1)^{2}}
\left[\max\left(\left\|\rho^{\Gamma}\right\|_1,\left\|\mathcal{R}
\left(\rho\right)\right\|_1\right)-1\right]^2
\end{equation}}
for $q\geq2$ with $d\geq3$  or $q\geq4$ with $d=2$. Additionally, we have the following bound
\begin{align}\label{corbound4}
\mathcal{C}^{t}_q\left(\rho\right)>\frac{1-2^{1-q}}{1-2^{1-s}}\left[\max
\left(\left\|\rho^{\Gamma}\right\|_1,\left\|\mathcal{R}
\left(\rho\right)\right\|_1\right)-1\right]^2
\end{align}
for $s\leq q<4$ with $d=2$.
\end{theorem}

$\mathit{Proof}$. For the pure state defined  in (\ref{ac1}), from  (\ref{bound3}) and (\ref{text006}), it follows that the function  $f(q)$ in (\ref{bound1}) is monotonically increasing with respect to $q$ for $q\geq2$ and $d\geq3$, or when $d=2$ and $q\geq s$.

For $d\geq3$, the following result holds
\begin{align}\label{theorem4}
\mathcal{C}^{t}_q\left(|\psi\rangle\right)&\geq\frac{d-d^{1-q}(1+(d-1)^{q})}
{d-d^{-1}(1+(d-1)^{2})}\mathcal{C}^{t}_2\left(|\psi\rangle\right) \nonumber\\
&\geq\frac{d-d^{1-q}(1+(d-1)^{q})}{d-d^{-1}(1+(d-1)^{2})}\frac{2}{d(d-1)}
\left(\|\delta^{\Gamma}\|_1-1\right)^2\nonumber\\
&=\frac{d-d^{1-q}(1+(d-1)^{q})}{(d-1)^{2}}
\left(\|\delta^{\Gamma}\|_1-1\right)^2,
\end{align}
for $q\geq2$, with $\delta=|\psi\rangle\langle\psi|$, the first inequality arises from the monotonicity of $f\left(q\right)$, while the second  follows from  (\ref{b2p}).

For $d=2$, following an approach analogous to (\ref{theorem4}), we derive the following result for $q\geq4$,
\begin{align}\label{theor5}
\mathcal{C}^{t}_q(|\psi\rangle)&\geq\frac{2-2^{2-q}}{2-2^{-2}}
\mathcal{C}^{t}_4(|\psi\rangle)\nonumber\\
&=\frac{4}{7}(2-2^{2-q})\mathcal{C}^{t}_4(|\psi\rangle)\nonumber\\
&= (2-2^{2-q})\mathcal{C}^{t}_2(|\psi\rangle)(\frac{8}{7}-\frac{\mathcal{C}^{t}_2(|\psi\rangle)}{7}) \nonumber\\
&\geq(2-2^{2-q})\mathcal{C}^{t}_2(|\psi\rangle)\nonumber\\
&\geq(2-2^{2-q})\left(\|\delta^{\Gamma}\|_1-1\right)^2,
\end{align}
where the second equality arises from the relation $f\left(4\right)=f(2)(\frac{8}{7}-\frac{f(2)}{7})$.

For the range $s\leq q<4$, we obtain the following
\begin{align}\label{case1}
\mathcal{C}^{t}_q(|\psi\rangle)&\geq\frac{2-2^{2-q}}{2-2^{2-s}}\mathcal{C}^{t}_s(|\psi\rangle)\nonumber\\
&>\frac{2-2^{2-q}}{2-2^{2-s}}\mathcal{C}^{t}_2(|\psi\rangle)\nonumber\\
&>\frac{2-2^{2-q}}{2-2^{2-s}}\left(\|\delta^{\Gamma}\|_1-1\right)^2\nonumber\\
&=\frac{1-2^{1-q}}{1-2^{1-s}}\left(\|\delta^{\Gamma}\|_1-1\right)^2,
\end{align}
here the second inequality follows from the monotonicity of $\mathcal{C}^{t}_q(|\psi\rangle)$  with respect to $q$.

Consider the optimal pure-state decomposition of $\rho$ for $\mathcal{C}^{t}_q\left(\rho\right)$, which can be expressed as $\rho=\sum_{i}p_i|\psi_i\rangle\langle\psi_i|$. Under the conditions $q\geq2$  when $d\geq3$ or $q\geq4$  when $d=2$, we establish the following result,
\begin{align}\label{text9}
\mathcal{C}^{t}_q\left(\rho\right)&=\sum_ip_i\mathcal{C}^{t}_q\left(|\psi_i\rangle\right)\nonumber\\
&\geq \frac{d-d^{1-q}(1+(d-1)^{q})}{(d-1)^{2}}\sum_ip_i\left(\left\|\rho_i^{\Gamma}\right\|_1-1\right)^2\nonumber\\
&\geq \frac{d-d^{1-q}(1+(d-1)^{q})}{(d-1)^{2}}\left(\sum_ip_i\left\|\rho_i^{\Gamma}\right\|_1-1\right)^2\nonumber\\
&\geq  \frac{d-d^{1-q}(1+(d-1)^{q})}{(d-1)^{2}}\left(\left\|\rho^{\Gamma}\right\|_1-1\right)^2,
\end{align}
here $\rho_i=|\psi_i\rangle\langle\psi_i|$. The first inequality is a direct result of  (\ref{theorem4}) and (\ref{theor5}), while the second from the convexity of the function $\mathit{f}\left(y\right)=y^2$. The final  follows from the convexity property of the trace norm, as well as from the condition $\|\rho^\Gamma\| \geq 1$ stated in  (\ref{text1}).

From (\ref{text1}), and following the reasoning similar to that in (\ref{theorem4}), (\ref{theor5}), and (\ref{text9}), we obtain the following bound
\begin{equation}\label{text10}
\mathcal{C}^{t}_q\left(\rho\right)\geq\frac{d-d^{1-q}(1+(d-1)^{q})}{(d-1)^{2}}
\left(\left\|\mathcal{R}\left(\rho\right)\right\|_1-1\right)^2,
\end{equation}
this holds for the cases where $q\geq2$ with $d\geq3$, and $q\geq4$ with $d=2$.$\hfill\qedsymbol$

Theorem \ref{thma} establishes an explicit connection between the $\mathcal{C}^{t}_q$-concurrence and two widely used entanglement criteria, the PPT and realignment criteria. It is seen that the bound becomes informative when $q \geq 2$ with $d \geq 3$ or $s\leq q<4$ with $d=2$. Corollary \ref{corollary} and Theorem \ref{thma} together provide a comprehensive framework for understanding and estimating the $\mathcal{C}^{t}_q$-concurrence. The former ensures the consistency and comparability of entanglement values across $q$, while the latter grounds $\mathcal{C}^{t}_q$ in operationally meaningful specifications. This double usefulness enhances the theoretical completeness and experimental accessibility of the overall consistency framework.

Theorem \ref{thma} presents the lower bounds for $\mathcal{C}^{t}_q\left(\rho\right)$.
In the following section, we proceed with the analytical calculation of the $\mathcal{C}^{t}_q$-concurrence for two specific types of quantum states, isotropic states and Werner states.

\section{$\mathcal{C}^{t}_q$-concurrence of Isotropic and Werner states}\label{sect4}
Let $\mathcal{C}$ represent a quantum entanglement measure, and let $\mathcal{S}$ be the set of all states, with $\mathcal{P}$ as the subset of pure states within $\mathcal{S}$. Consider $\mathcal{G}$ to be a compact group acting on $\mathcal{S}$ through the transformation $\left(U, \varrho\right) \mapsto U \varrho U^\dagger$. We assume that the entanglement measure $\mathcal{C}$ is invariant under the action of $\mathcal{G}$ on the pure states within $\mathcal{P}$. We can define a projection operator $\mathbf{P}: \mathcal{S} \to \mathcal{S}$ such that $\mathbf{P} \varrho = \int dU \, U \varrho U^\dagger$, where $dU$ represents the normalized Haar measure on the group $\mathcal{G}$. Additionally, we introduce the function $\zeta$ defined on the set $\mathbf{P}\mathcal{S}$ as follows \cite{Vollbrecht062307}
\begin{equation}\label{p30}
\zeta\left(\varrho\right)=\min\left\{\mathcal{C}\left(\ket{\phi}\right):\ket{\phi}\in \mathcal{P},\,\mathbf{P}|\phi\rangle\langle\phi|=\varrho\right\}.
\end{equation}

For any state $\varrho \in \mathbf{P}S$, the entanglement measure $\mathcal{C}(\varrho)$ can then be expressed as
\begin{equation}\label{p31}
\mathcal{C}\left(\varrho\right)=co\left(\zeta\left(\varrho\right)\right),
\end{equation}
where $co(\zeta)$ denotes the convex-roof extension of the function $\zeta$, which is defined as the convex hull of $\zeta$.
\subsection{Isotropic states}

Isotropic states, denoted as $\varrho_\mathcal{F}$\cite{Horodecki4206}. Specifically, for a given fidelity parameter $\mathcal{F}$, an isotropic state can be expressed as
\begin{equation}
\varrho_\mathcal{F}=\frac{1-\mathcal{F}}{d^2-1}\left(I-|\Phi^+\rangle\langle\Phi^+|\right)
+\mathcal{F}|\Phi^+\rangle\langle\Phi^+|,
\end{equation}
The maximally entangled pure state is represented by $\ket{\Phi^+}=\frac{1}{\sqrt{d}}\sum_{k=1}^d\ket{kk}$, where $I$  denotes the identity operator. The fidelity of the state $\varrho_\mathcal{F}$ with respect to $\ket{\Phi^+}$, denoted as $\mathcal{F}$, is defined by $\mathcal{F}=\braket{\Phi^+|\varrho_\mathcal{F}|\Phi^+}$, where $\mathcal{F}$ lies in the range $0 \leq \mathcal{F} \leq 1$. The isotropic state $\varrho_\mathcal{F}$ is separable when $\mathcal{F}\leq 1/d$.

Additionally, the state $\varrho_\mathcal{F}$ remains invariant under the transformation $ V \otimes V^* $ for any unitary operator $V $, as demonstrated in \cite{Horodecki4206}. The entanglement characteristics of isotropic states have been thoroughly investigated, with several measures being explored, including the EoF \cite{Terhal2625}, tangle and concurrence \cite{Rungta012307}, and the R$\mathrm{\acute{e}}$nyi \( \alpha \)-entropy entanglement \cite{Wang022324}. Furthermore, it has been established that when $\mathcal{F} > 1/d$, the trace norm of the partial transpose $\varrho_\mathcal{F}^\Gamma $ and the realigned matrix $\mathcal{R}(\varrho_\mathcal{F})$ are identical, such that $ \|\varrho_\mathcal{F}^\Gamma\|_1 = \|\mathcal{R}(\varrho_\mathcal{F})\|_1 = d\mathcal{F} $ \cite{Rudolph219,Vidal032314}.

Building upon the methods introduced in \cite{Terhal2625, Rungta012307}, we proceed to derive the expression for the $\mathcal{C}^{t}_q$-concurrence of the $\varrho_\mathcal{F}$, with a detailed derivation provided in Appendix \ref{b}. The corresponding expression for the $\mathcal{C}^{t}_q$-concurrence is as follows
\begin{equation}
\mathcal{C}^{t}_q\left(\varrho_\mathcal{F}\right)=co\left(\zeta\left(\mathcal{F},q,d\right)\right),
\end{equation}
here, $co\left(\zeta\left(\mathcal{F},q,d\right)\right)$ denotes the largest convex function that is bounded above by $\zeta\left(\mathcal{F},q,d\right)$, which can be explicitly expressed as
\begin{equation}\label{example10}
\zeta\left(\mathcal{F},q,d\right)=d-(\chi^{2q}+(1-\chi^{2})^{q})-\left(d-1\right)
(\sigma^{2q}+(1-\sigma^{2})^{q}),
\end{equation}
Here, the parameters $ \chi $  and $  \sigma $  are defined as
$\chi=\frac{\sqrt{\mathcal{F}}+\sqrt{\left(d-1\right)\left(1-\mathcal{F}\right)}}{\sqrt{d}}$ and $\sigma=\frac{1}{\sqrt{d}}(\sqrt{\mathcal{F}}-\frac{\sqrt{1-\mathcal{F}}}{\sqrt{d-1}})$
as derived in \cite{Yang052423}.

Next, we demonstrate the tightness of the lower bound given in (\ref{text100}) by applying it to isotropic states. We begin by considering the specific case where $q=3$.

%Next, we illustrate the tightness of the lower bound (\ref{text100}) by the isotropic states. Let us first consider the case of $q=3$.
%To show the tightness of our lower bound (\ref{text100}), let us first consider the case of $q=3$.

(i) In the special case where $d=2$, the expression in (\ref{example10}) reduces to the following form
\begin{align}\label{example2}
\zeta\left(\mathcal{F},3,2\right)=\left(2\mathcal{F}-1\right)^2,
\end{align}
which is valid for $\mathcal{F}$ within the interval $(\frac{1}{2},1]$.

Since the $\frac{\partial^2}{\partial \mathcal{F}^2} \zeta(\mathcal{F}, 3, 2) \geq 0$, it follows that
\begin{equation}\label{example3}
\mathcal{C}^{t}_3\left(\varrho_\mathcal{F}\right)=\begin{cases}
0,       & \mathcal{F}\leq \frac{1}{2},\\[2mm]
\displaystyle\left(2\mathcal{F}-1\right)^2,  &\frac{1}{2}< \mathcal{F}\leq 1,
\end{cases}
\end{equation}
which directly corresponds to the lower bound given in (\ref{text100}).
As a result, for the case where $q=3$ and $d=2$, the lower bound provided in (\ref{text100}) coincides with the exact value of the $\mathcal{C}^{t}_q$-concurrence for any two-qubit isotropic state $\varrho_\mathcal{F}$.
%Therefore, for $q=3$ and $d=2$ our lower bound (\ref{text100}) is just the exact value of the $\mathcal{C}^{t}_q$-concurrence for any two-qubit isotropic state $\rho_F$.

%In fact, from (\ref{bound1}) we have $f\left(3\right)=2f\left(2\right)$ for $d=2$. Hence, $\mathcal{C}^{t}_3\left(\rho_F\right)=2\mathcal{C}^{t}_2\left(\rho_F\right)$, which is in consistent with the $\mathcal{C}^{t}_2$-concurrence of any two-qubit isotropic state $\rho_F$ \cite{PhysRevA.67.012307}. This implies that our lower bound (\ref{text100}) is exact for both $\mathcal{C}^{t}_2\left(\rho_F\right)$ and $\mathcal{C}^{t}_3\left(\rho_F\right)$ when $d=2$.

In fact, from (\ref{bound1}), it follows that for $d=2$, the function satisfies the relation $f\left(3\right)=2f\left(2\right)$. Consequently, $\mathcal{C}^{t}_3\left(\varrho_\mathcal{F}\right)=2\mathcal{C}^{t}_2\left(\varrho_\mathcal{F}\right)$,
which is consistent with the previously established expression for the $\mathcal{C}^{t}_2$-concurrence of $\varrho_\mathcal{F}$ \cite{Rungta012307}. This result confirms that the lower bound given in (\ref{text100}) is exact for both $\mathcal{C}^{t}_2\left(\varrho_\mathcal{F}\right)$ and $\mathcal{C}^{t}_3\left(\varrho_\mathcal{F}\right)$ when $d=2$.

(ii)  For the case of $d=3$,  (\ref{example10}) simplifies to
\begin{equation}\label{text0331}
\zeta\left(\mathcal{F},3,3\right)=3-(\chi^{6}+(1-\chi^{2})^{3})-2(\sigma^{6}+(1-\sigma^{2})^{3}),
\end{equation}
for any $\mathcal{F}\in(1/3,1]$, where the parameters $\chi $ and $ \sigma $ are given by
$\chi=\frac{\sqrt{\mathcal{F}}+\sqrt{2\left(1-\mathcal{F}\right)}}{\sqrt{3}}$ and
$\sigma=\frac{1}{\sqrt{3}}(\sqrt{\mathcal{F}}-\frac{\sqrt{1-\mathcal{F}}}{\sqrt{2}})$.

Given that $\frac{\partial}{\partial \mathcal{F}} \zeta(\mathcal{F}, 3, 3) > 0$, the function $\zeta\left(\mathcal{F},3,3\right)$ increases monotonically in the region where $\varrho_\mathcal{F}$  maintains entanglement. This trend is illustrated in FIG. \ref{f3}. When $\mathcal{F}\geq 0.94$, the $\frac{\partial^2}{\partial \mathcal{F}^2} \zeta(\mathcal{F}, 3, 3) < 0$, indicating that $\zeta\left(\mathcal{F},3,3\right)$ loses its convexity in the range $\mathcal{F}\in[0.94,1]$. Since $\mathcal{C}^{t}_3\left(\varrho_\mathcal{F}\right)$ is defined as the largest convex function that remains bounded above by (\ref{text0331}), we ensure convexity by linearly interpolating between the points at  $\mathcal{F}=0.94$ and $\mathcal{F}=1$.

%Since the second derivative of $\zeta\left(F,3,3\right)$ with respect to $F$ becomes negative when $F\geq 0.94$, $\zeta\left(F,3,3\right)$ is no longer a convex function for $F\in[0.94,1]$. As $\mathcal{C}^{t}_3\left(\rho_F\right)$ is the largest convex function that is upper bounded by (\ref{text0331}), we connect the point $F=0.94$ with the point $F=1$ by a straight line.
Thus, we obtain the following result, as illustrated in FIG. \ref{f3},
\begin{equation}\label{c333}
\mathcal{C}^{t}_3\left(\varrho_\mathcal{F}\right)=\begin{cases}
0,       & \mathcal{F}\leq 1/3,\\
\zeta\left(\mathcal{F},3,3\right), &1/3< \mathcal{F}\leq0.94,\\
2.23\mathcal{F}-1.23  ,  &0.94< \mathcal{F}\leq 1.
\end{cases}
\end{equation}
From Theorem \ref{thma}, we derive the following lower bound
\begin{equation}\label{exc33}
\mathcal{C}^{t}_3\left(\varrho_\mathcal{F}\right)\geq\frac{\left(3\mathcal{F}-1\right)^2 }{4}.
\end{equation}
As depicted in FIG. \ref{d3q3}, this lower bound in (\ref{exc33}) is shown to be tight.
\begin{figure}[htbp]
    \begin{minipage}[t]{0.9\linewidth}
        \centering
        \includegraphics[width=\textwidth]{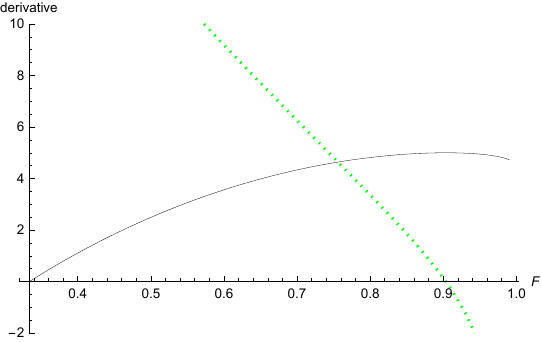}
        \caption{ The plot illustrates the first (solid black line) and second (dashed green line) derivatives of $\zeta\left(\mathcal{F},3,3\right)$ with respect to $\mathcal{F}$. The solid black line represents the first derivative, while the dashed green line shows the second derivative of the function.}
    \label{f3}
    \end{minipage}%
    \hspace{.20in}
    \begin{minipage}[t]{0.9\linewidth}
        \centering
        \includegraphics[width=\textwidth]{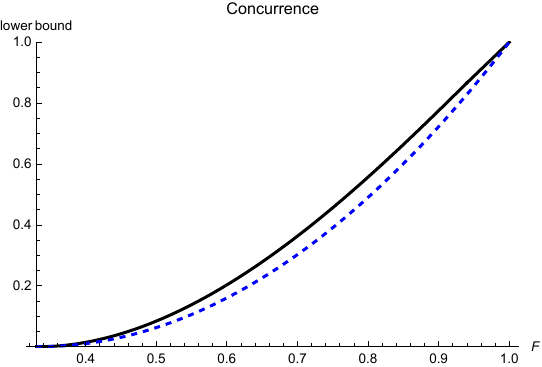}
   \caption{ The plot shows $\mathcal{C}^{t}_3\left(\varrho_\mathcal{F}\right)$ versus $\mathcal{F}$  for $d=3$. The solid black line represents the exact value given by equation (\ref{c333}), while the dashed blue line represents the lower bound given by equation (\ref{exc33}).}
    \label{d3q3}
    \end{minipage}%
    \end{figure}

For $q=4$ in the scenario where $d=2$, we obtain the following result,
\begin{align}\label{example4}
\zeta\left(\mathcal{F},4,2\right)=2-(\chi^{8}+(1-\chi^{2})^{4})-(\sigma^{8}+(1-\sigma^{2})^{4})
\end{align}
for any $\mathcal{F}\in(\frac{1}{2},1]$. Since the $\frac{\partial^2}{\partial \mathcal{F}^2} \zeta(\mathcal{F}, 4, 2) \geq 0$, it follows that
\begin{equation}\label{example5}
\mathcal{C}^{t}_4\left(\varrho_\mathcal{F}\right)=\begin{cases}
0,       & \mathcal{F}\leq \frac{1}{2},\\[2mm]
\displaystyle\frac{7+4\mathcal{F}(1-\mathcal{F})}{7}\left(2\mathcal{F}-1\right)^2,  &\frac{1}{2}< \mathcal{F}\leq 1.
\end{cases}
\end{equation}

From Theorem \ref{thma}, we derive the following lower bound
\begin{equation}\label{example42}
\mathcal{C}^{t}_4\left(\varrho_\mathcal{F}\right)\geq \left(2\mathcal{F}-1\right)^2.
\end{equation}
As shown in FIG. \ref{d2q4}, this lower bound in (\ref{example42}) is proven to be tight.

For $q=4$ and $d=3$,  (\ref{example10}) reduces to
\begin{equation}\label{text031}
\zeta\left(\mathcal{F},4,3\right)=3-(\chi^{8}+(1-\chi^{2})^{4})-2(\sigma^{8}+(1-\sigma^{2})^{4})
\end{equation}
for any $\mathcal{F}\in(1/3,1]$, the parameters $\chi $ and $ \sigma $ are defined as
$\chi=\frac{\sqrt{\mathcal{F}}+\sqrt{2\left(1-\mathcal{F}\right)}}{\sqrt{3}}$ and $\sigma=\frac{1}{\sqrt{3}}(\sqrt{\mathcal{F}}-\frac{\sqrt{1-\mathcal{F}}}{\sqrt{2}})$.
Since the $\frac{\partial}{\partial \mathcal{F}} \zeta(\mathcal{F}, 4, 3) > 0$, the function $\zeta\left(\mathcal{F},4,3\right)$ exhibits a monotonically increasing trend within the range where $\varrho_\mathcal{F}$ remains entangled. This behavior is illustrated in FIG. \ref{f4}.
%for any $F\in(1/3,1]$, where $\gamma=\frac{\sqrt{F}+\sqrt{2\left(1-F\right)}}{\sqrt{3}}$ and$\delta=\frac{1}{\sqrt{3}}(\sqrt{F}-\frac{\sqrt{1-F}}{\sqrt{2}})$. As the first derivative of $\zeta\left(F,4,3\right)$ with respect to $F$ is always positive, $\zeta\left(F,4,3\right)$ is monotonically increasing in the regime where $\rho_F$ is entangled, see FIG. \ref{f4}.
\begin{figure}[htbp]
    \begin{minipage}[t]{0.9\linewidth}
        \centering
        \includegraphics[width=\textwidth]{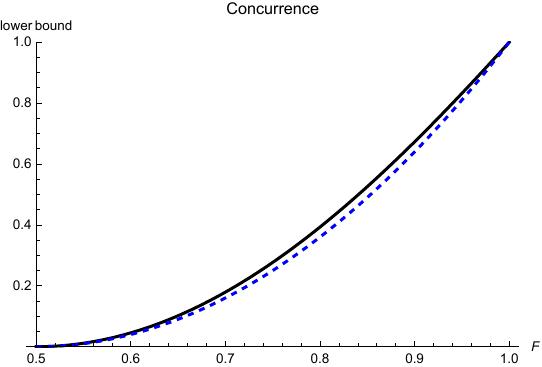}
        \caption{The plot of $\mathcal{C}^{t}_4\left(\varrho_\mathcal{F}\right)$ versus $\mathcal{F}$ for $d=2$ is as follows: the solid black line represents the expression in (\ref{example5}), while the dashed blue line corresponds to the lower bound given by (\ref{example42}).}
  \label{d2q4}
    \end{minipage}%
    \hspace{.20in}
    \begin{minipage}[t]{0.9\linewidth}
        \centering
        \includegraphics[width=\textwidth]{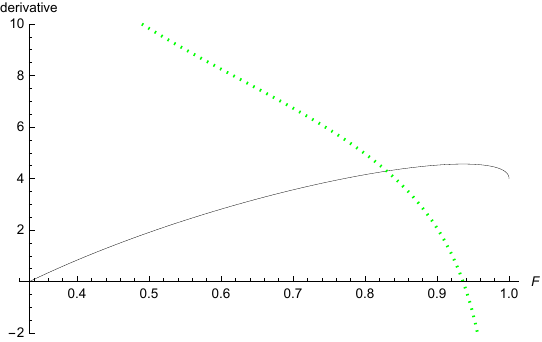}
   \caption{The first (solid black line) and second (dashed green line) derivatives of $\zeta\left(\mathcal{F},4,3\right)$ with respect to $\mathcal{F}$ are shown in the plot. The solid black line represents the first derivative, while the dashed green line represents the second derivative.}
  \label{f4}
    \end{minipage}%
\end{figure}
When $\mathcal{F}\geq 0.904$, the $\frac{\partial^2}{\partial \mathcal{F}^2} \zeta(\mathcal{F}, 4, 3) < 0$, indicating that $\zeta\left(\mathcal{F},4,3\right)$ loses its convexity in the interval $\mathcal{F}\in[0.904,1]$. To maintain the convexity of $\mathcal{C}_4\left(\varrho_\mathcal{F}\right)$, which is the largest convex function bounded above by (\ref{text031}), we connect the points at $\mathcal{F} = 0.904$ and $\mathcal{F} = 1$ with a straight line.
As a result, the curve presented in FIG. \ref{f4} is obtained,
\begin{equation}\label{c33}
\mathcal{C}^{t}_4\left(\varrho_\mathcal{F}\right)=\begin{cases}
0,       & \mathcal{F}\leq 1/3,\\
\zeta\left(\mathcal{F},4,3\right), &1/3< \mathcal{F}\leq0.904,\\
2.0658 \mathcal{F}-1.06566,  &0.904< \mathcal{F}\leq 1.
\end{cases}
\end{equation}

%From Theorem \ref{thma} we get that\begin{equation}\label{exc3}\mathcal{C}^{t}_4\left(\rho_F\right)\geq\frac{\left(3F-1\right)^2 }{4}.\end{equation}Obviously, it is seen from FIG. \ref{d3q4} that our lower bound of (\ref{exc3}) is tight.

By utilizing Theorem \ref{thma}, we derive the following lower bound for the $\mathcal{C}_4^t$-concurrence of the  $\varrho_\mathcal{F}$,
\begin{equation}\label{exc3}
\mathcal{C}^{t}_4\left(\varrho_\mathcal{F}\right)\geq\frac{\left(3\mathcal{F}-1\right)^2 }{4}.
\end{equation}
As clearly illustrated in FIG. \ref{d3q4}, this lower bound is tight.

\begin{figure}[htbp]
    \begin{minipage}[t]{0.9\linewidth}
        \centering
        \includegraphics[width=\textwidth]{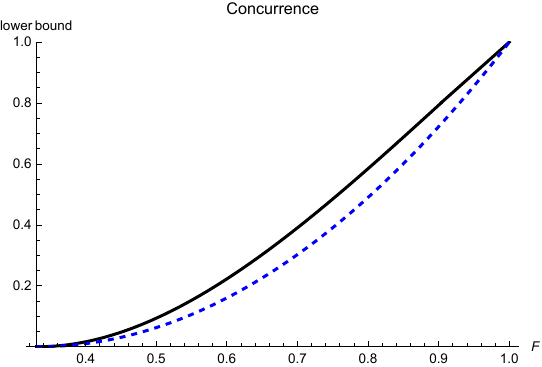}
        \caption{The plot of $\mathcal{C}^{t}_4\left(\varrho_\mathcal{F}\right)$ versus $\mathcal{F}$ for $d=3$ shows the following: the solid black line represents the expression given by (\ref{c33}), while the dashed blue line corresponds to the lower bound of (\ref{exc3}).}
  \label{d3q4}
    \end{minipage}%
    \hspace{.20in}
    \begin{minipage}[t]{0.9\linewidth}
        \centering
        \includegraphics[width=\textwidth]{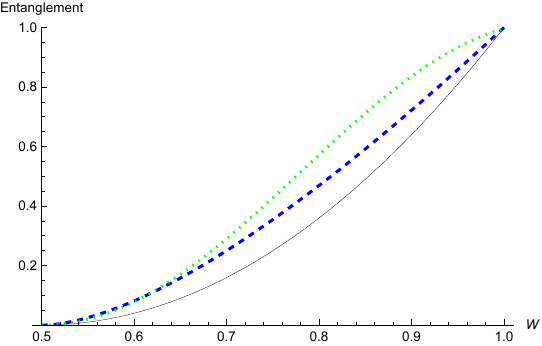}
   \caption{For the case of Werner states, the plot displays the following entanglement measures: the solid (black) line represents the $\mathcal{C}^{t}_2$-concurrence, $\mathcal{C}^{t}_2\left(\varrho_\textit{w}\right)$; the dashed (blue) line corresponds to the EoF, $E_F\left(\varrho_\textit{w}\right)$; and the dot-dashed (green) line illustrates the $\mathcal{C}^{t}_8$-concurrence, $\mathcal{C}^{t}_8\left(\varrho_\textit{w}\right)$.}
  \label{f2}
    \end{minipage}%
    \end{figure}

\subsection{Werner states}
For Werner states, the general expressions are
\begin{align}\label{p4wer1}
\varrho_\textit{w}=&\frac{2\left(1-\textit{w}\right)}{d\left(d+1\right)}\left(\sum_{i=1}^d|ii\rangle\langle ii|+\sum_{l<k}|\Phi_{lk}^+\rangle\langle\Phi_{lk}^+|\right)\nonumber\\
&+\frac{2\textit{w}}{d\left(d-1\right)}\sum_{l<k}|\Phi_{lk}^-\rangle\langle\Phi_{lk}^-|,
\end{align}
the states $\ket{\Phi_{lk}^{\pm}}$ are defined as $\ket{\Phi_{lk}^{\pm}}=(\ket{lk}\pm\ket{kl})/\sqrt{2}$, where $\ket{lk}$ and $\ket{kl}$ represent the standard basis states in the Hilbert space of the two subsystems.
The mixing parameter $\textit{w}$ is determined by the trace of the projection operator onto the antisymmetric subspace and is explicitly given by
$\textit{w}=\mathrm{tr}(\varrho_\textit{w}\sum_{l<k}|\Phi_{lk}^-\rangle\langle\Phi_{lk}^-|)$
as established in \cite{Lee062304}.  $\varrho_\textit{w}$ is separable iff $0\leq \textit{w}\leq \frac{1}{2}$ \cite{Vollbrecht062307,Werner4277}.
For values of $\textit{w}>\frac{1}{2}$, the state exhibits entanglement, and its entanglement measure can be explicitly formulated as follows, with a detailed derivation provided in Appendix \ref{c},
\begin{equation}\label{text32wer}
\zeta\left(\varrho_\textit{w}\right)
=2\left(1-\left(\frac{1+G}{2}\right)^{q}
-\left(\frac{1-G}{2}\right)^{q}\right)
\end{equation}
where $G=2\sqrt{\textit{w}(1-\textit{w})}$, for any $q\geq2$. This expression gives the entanglement of the Werner state for $\textit{w}\leq\frac{1}{2}$, which depends on both $\textit{w}$ and the power $q$.

Let us begin by considering the specific case where $q=3$ and $d=2$. In this case, the expression (\ref{text32wer}) simplifies to
\begin{align}\label{example2}
\zeta\left(\varrho_\textit{w}\right)=(2\textit{w}-1)^2,
\end{align}
which holds for $\textit{w}\leq\frac{1}{2}$.

Since the $\frac{\partial^2}{\partial w^2} \zeta(\varrho_w) \geq 0$, it follows that
\begin{equation}\label{example3}
\mathcal{C}^{t}_3\left(\varrho_\textit{w}\right)=\begin{cases}
0,       & \textit{w}\leq \frac{1}{2},\\[2mm]
\displaystyle\left(2\textit{w}-1\right)^2,  &\textit{w}\leq \frac{1}{2},
\end{cases}
\end{equation}
this behavior corresponds exactly to our lower bound given in equation (\ref{text100}). As a result, for the case where \( q = 3 \) and \( d = 2 \), the lower bound expressed in (\ref{text100}) exactly matches the value of the \( \mathcal{C}^{t}_q \)-concurrence for the  \( \varrho_\textit{w} \).
This confirms that our bound is tight in this scenario.
%Therefore, for $q=3$ and $d=2$ our lower bound (\ref{text100}) is just the exact value of the $\mathcal{C}^{t}_q$-concurrence for  two-qubit Werner state $\rho_W$.
%The concurrence of Werner states has been obtained in \cite{PhysRevA62044302}, $C\left(\rho_W\right)=2W-1$ for $W>1/2$. It is direct to find that  $\mathcal{C}^{t}_q\left(\rho_W\right)=C^2\left(\rho_W\right)$ when $q=2$. Moreover, the EoF for Werner states is given by \cite{PhysRevA.64.062307}, $E_F\left(\rho_W\right)=h[\frac{1+\sqrt{1-C^2}}{2}]$, where $h[x]=-x\log_2 x-(1-x)\log_2(1-x)$. Therefore, we easily obtain the relationship between $E_F\left(\rho_W\right)$ and $\mathcal{C}^{t}_2\left(\rho_W\right)$, $E_F\left(\rho_W\right)=h[\frac{1+\sqrt{1-\mathcal{C}^{t}_2}}{2}]$.

%In FIG. \ref{f2} we illustrate the relations among the EoF and the $\mathcal{C}^{t}_q$-concurrences of Werner states for $q=2$ and $q=8$. The EoF for Werner states $E_F\left(\rho_W\right)$ is always upper bounded by the $\mathcal{C}^{t}_8$-concurrence $\mathcal{C}^{t}_8\left(\rho_W\right)$, and larger than the $\mathcal{C}^{t}_2$-concurrence $\mathcal{C}^{t}_2\left(\rho_W\right)$ for $W>0.5$.
The concurrence of Werner states was derived in \cite{Hiroshima044302}, where it is given by $C\left(\varrho_\textit{w}\right)=2\textit{w}-1$  for $\textit{w}>1/2$. It is straightforward to observe that for $q=2$,  $\mathcal{C}^{t}_q\left(\varrho_\textit{w}\right)=C^2\left(\varrho_\textit{w}\right)$. Additionally, the EoF for Werner states is provided in \cite{Vollbrecht062307} as $E_F\left(\varrho_\textit{w}\right)=h[\frac{1+\sqrt{1-C^2}}{2}]$, where $h[x]=-x\log_2 x-(1-x)\log_2(1-x)$ is the binary entropy function. This gives rise to the relationship between $E_F\left(\varrho_\textit{w}\right)$ and $\mathcal{C}^{t}_2\left(\varrho_\textit{w}\right)$, specifically, $E_F\left(\varrho_\textit{w}\right)=h[\frac{1+\sqrt{1-\mathcal{C}^{t}_2}}{2}]$.
In FIG. \ref{f2}, we show the relationships between the EoF and the $\mathcal{C}^{t}_q$-concurrences of Werner states for $q=2$ and $q=8$. The EoF $E_F\left(\varrho_\textit{w}\right)$ is always bounded above by the $\mathcal{C}^{t}_8$-concurrence $ \mathcal{C}^{t}_8\left(\varrho_\textit{w}\right)$, and is greater than the $\mathcal{C}^{t}_2$-concurrence $ \mathcal{C}^{t}_2\left(\varrho_\textit{w}\right)$ when $\textit{w}>0.5$.

\section{ Monogamy inequality of the $\mathcal{C}^{t}_q $-concurrence in multipartite systems}\label{sect5}

The distribution of entanglement is a significant issue in entanglement theory and quantum information processing. Monogamy relations highlight a fundamental aspect of connections among a specific subset of bipartite entanglements within a multiparty system. The COFFMAN052306 inequality proves quantitatively the original monogamy of entanglement \cite{COFFMAN052306}.

%Consider any ${\cal C}^2\otimes {\cal C}^d$ pure state $|\phi\rangle_{AB}$ with Schmidt decomposition, $|\phi\rangle_{AB}=\sqrt{\lambda_0}|0\rangle|\phi_0\rangle+\sqrt{\lambda_1}|1\rangle|\phi_1\rangle$,where the subsystem $A$ is a qubit one while the subsystem $B$ is a finite dimensional space, $\{|\phi_i\rangle\}$ are orthogonal states on subsystem $B$.

Consider a pure state $|\phi\rangle_{AB}$ in the bipartite system ${\cal C}^2\otimes {\cal C}^d$, where subsystem $A$ is a qubit $( \mathcal{C}^2)$ and subsystem $ B$ is a finite-dimensional space $( \mathcal{C}^d)$. The state can be written in its Schmidt decomposition form as
$|\phi\rangle_{AB}=\sqrt{\lambda_0}|0\rangle|\phi_0\rangle+\sqrt{\lambda_1}|1\rangle|\phi_1\rangle$,
where $\lambda_0 $ and $\lambda_1$ are the Schmidt coefficients, and $ \{|\phi_i\rangle\} $ are orthogonal states in subsystem $B$ (i.e., $\langle \phi_i | \phi_j \rangle = \sigma_{ij}$).
%Here, \( |0\rangle \) and \( |1\rangle \) form an orthonormal basis for subsystem \( A \), and \( |\phi_0\rangle \) and \( |\phi_1\rangle \) are orthogonal states in subsystem \( B \). The coefficients \( \lambda_0 \) and \( \lambda_1 \) satisfy \( \lambda_0 + \lambda_1 = 1 \), and they determine the amount of entanglement between the two subsystems.
From the expression in (\ref{Cq1}), we can derive the $\mathcal{C}^{t}_q$-concurrence for the pure state $ |\phi\rangle_{AB}$, which is given by
\begin{eqnarray}
\mathcal{C}^{t}_q(|\phi\rangle_{AB})=\frac{1}{\mu}{C}^t_q(\varrho_A),\nonumber
\end{eqnarray}
here, $\varrho_A={\rm tr}_B(|\phi\rangle_{AB}\langle\phi|)$, while $\mu=2-2^{2-q}$ serves as the normalization factor.
Explicitly, we have
\begin{eqnarray}
\mathcal{C}^{t}_q(|\phi\rangle_{AB})=\frac{1}{1-2^{1-q}}(1-{\lambda_0}^q-{\lambda_1}^q)\nonumber,
\end{eqnarray}
where $\lambda_0$ and $\lambda_1$ are the Schmidt coefficients from the decomposition of $|\phi\rangle_{AB}$. This expression shows that the $\mathcal{C}^{t}_q $-concurrence depends on the Schmidt coefficients $\lambda_0$ and $\lambda_1$, as well as the value of \( q \), which parameterizes the degree of the $\mathcal{C}^{t}_q$-concurrence. The factor $\mu$ ensures the normalization of the concurrence.

%From (\ref{Cq1}), we have the $\mathcal{C}^{t}_q$-concurrence,\begin{eqnarray}\mathcal{C}^{t}_q(|\phi\rangle_{AB})&=&\frac{1}{\mu}{C}^t_q(\varrho_A)\nonumber\\&=&\frac{1}{1-2^{1-q}}(1-{\lambda_0}^q-{\lambda_1}^q)\nonumber,\label{hx1}\end{eqnarray}where $\varrho_A={\rm tr}_B(|\Phi\rangle_{AB}\langle\Phi|)$, and $\mu=2-2^{2-q}$ is a normalized factor.

%Besides, the concurrence of $|\phi\rangle_{AB}$ is given by$C(|\phi\rangle_{AB})=\sqrt{2(1-{\rm{tr}}(\varrho_A)^2)}=2\sqrt{\lambda_0\lambda_1}$.We can find a one-to-one correspondence between the Schmidt coefficients $\lambda_0, \lambda_1$ of  $|\phi\rangle_{AB}$ and $C(|\phi\rangle_{AB})$, $\lambda_{0,1}=\frac{1\pm\sqrt{1-C^2(|\phi\rangle_{AB})}}{2}$.

Furthermore, the concurrence associated with the pure state $|\phi\rangle_{AB}$ is expressed as
$C(|\phi\rangle_{AB})=\sqrt{2(1-{\rm{tr}}(\varrho_A)^2)}=2\sqrt{\lambda_0\lambda_1}$,
where $\varrho_A = \text{tr}_B(|\phi\rangle_{AB} \langle \phi|)$, and $\lambda_0, \lambda_1$ correspond to the Schmidt coefficients. A one-to-one correspondence exists between the Schmidt coefficients $\lambda_0, \lambda_1$ and the concurrence $C(|\phi\rangle_{AB})$. The Schmidt coefficients can be expressed as
$\lambda_{0,1}=\frac{1\pm\sqrt{1-C^2(|\phi\rangle_{AB})}}{2}$.

It can be demonstrated that the $\mathcal{C}^{t}_q$-concurrence for any pure state $|\phi\rangle_{AB}$ satisfies the following relationship
\small{\begin{eqnarray}
\mathcal{C}^{t}_q(|\phi\rangle_{AB})=h_q(C(|\phi\rangle_{AB})),
\label{relation0}
\end{eqnarray}}
where $h_q(x)$ is an analytic function given by
\small{
\begin{eqnarray}
h_q(x):=\frac{1}{1-2^{1-q}}\left[ 1-\left(\frac{1+\sqrt{1-x^2}}{2}\right)^q-\left(\frac{1-\sqrt{1-x^2}}{2}\right)^q\right],
\label{hq(x)}
\end{eqnarray}}
with $0 \leq x \leq 1$. It has been established that the function $h_q(x)$ is monotonically increasing and convex for $ x\in (0, 1)$ when $2\leq q\leq 4$  \cite{Kim062328}. This property implies that for any mixed state $\rho_{AB} $ in the form ${\cal C}^2 \otimes {\cal C}^d $, the $\mathcal{C}^{t}_q$-concurrence  exhibits the same functional dependence as in  (\ref{relation0}).

%It has been verified that$h_q(x)$ is a monotonically increasing and convex function on $ x\in (0, 1)$ for any $2\leq q\leq 4$ \cite{PhysRevA.81.062328}. Based on this property, we can show that for any ${\cal C}^2\otimes {\cal C}^d$ mixed state, the $\mathcal{C}^{t}_q$-concurrence has the same relation as (\ref{relation0}).
\begin{theorem}\label{thmhqx}
For any qubit-qudit state $\rho_{AB}$, and for any $2\leq q\leq 4$, the $\mathcal{C}^{t}_q$-concurrence is related to the concurrence $C(\rho_{AB})$ by the following relationship
\begin{equation}
\mathcal{C}^{t}_q(\rho_{AB})=h_q(C(\rho_{AB})),
\label{thmhqx1}
\end{equation}
where $C(\rho_{AB})$ denotes the concurrence of the state $\rho_{AB}$, and $h_q(x)$ is the function given by (\ref{hq(x)}). This relation allows us to compute the $\mathcal{C}^{t}_q$-concurrence based on the standard concurrence for any qubit-qudit state.
\end{theorem}

$\mathit{Proof}$. For any given mixed state $\rho_{AB}$ of the form ${\cal C}^2\otimes {\cal C}^d$, we can express it as an optimal pure state decomposition $\rho_{AB}=\sum_kp_k|\nu_k\rangle_{AB}\langle\nu_k|$, where $p_k$ are the probabilities of the pure states $|\nu_k\rangle_{AB}$, and these pure states form the optimal decomposition of$\rho_{AB}$. In this case, the concurrence of the mixed state $C(\rho_{AB})$ is related to the concurrences of the pure states $ |\nu_k\rangle_{AB}$ as follows
$C(\rho_{AB})=\sum_kp_kC(|\nu_k\rangle_{AB})$. We can get
\begin{eqnarray}
h_q(C(\rho_{AB}))
&=&h_q\left(\sum_kp_kC(|\nu_{k}\rangle_{AB})\right)
\nonumber\\
&=&\sum_kp_kh_q(C(|\nu_{k}\rangle_{AB}))\nonumber\\
&=&\sum_kp_k\mathcal{C}^{t}_q(|\nu_{k}\rangle_{AB})
\label{eqnEC03}\\
&\geq&\mathcal{C}^{t}_q(\rho_{AB}),
\label{eqnEC2}
\end{eqnarray}
where the equality (\ref{eqnEC03}) is derived  using the relation established in (\ref{relation0}),  the inequality (\ref{eqnEC2}) follows directly from the definition of the $\mathcal{C}^{t}_q(\rho_{AB})$-concurrence, as given in (\ref{Cq2}).

Let the mixed state $\rho_{AB}$ be expressed as the optimal pure state decomposition, given by $\rho_{AB}=\sum_lp_l|u_l\rangle_{AB}\langle u_l|$, where $ p_l$ represents the probability associated with each pure state $|u_l\rangle_{AB} $. This decomposition is used to compute the $\mathcal{C}^{t}_q$-concurrence of $\rho_{AB}$. So we get
\begin{eqnarray}
\mathcal{C}^{t}_q(\rho_{AB})
&=&\sum_lp_l\mathcal{C}^{t}_q(|u_l\rangle_{AB})\nonumber\\
&=&\sum_lp_lh_q(C(|u_l\rangle_{AB}))\label{eqnEC3}\\
&\geq& h_q\left(\sum_lp_lC(|u_l\rangle_{AB})\right)\label{eqnEC4}\\
&\geq& h_q(C(\rho_{AB})),\label{eqnEC5}
\end{eqnarray}
where the equality in (\ref{eqnEC3}) follows directly from the relation (\ref{relation0}), while the inequality in (\ref{eqnEC4}) arises from the convexity property of the $h_q(x)$. Additionally, the inequality in (\ref{eqnEC5}) is a consequence of the definition of the concurrence in (\ref{concur}).
By combining (\ref{eqnEC2}) with (\ref{eqnEC5}), we have  completed the proof.$\hfill\qedsymbol$

It has been demonstrated that the squared concurrence follows a monogamy relation for any $n$-qubit state $\rho_{A_1A_2\cdots A_k}$ \cite{Osborne220503}, which is expressed as:
\begin{eqnarray}
C^2(\rho_{A_1|A_2\cdots A_k}) \geq \sum_{i=2}^k C^2(\rho_{A_1|A_i}),
\label{squaredconcurrence}
\end{eqnarray}
where $C(\rho_{A_1|A_2\cdots A_k})$ represents the bipartite entanglement for the bipartition between subsystem $A_1$ and the remaining subsystems $A_2 \cdots A_k$. On the other hand, $C(\rho_{A_1|A_i})$ denotes the bipartite entanglement corresponding to the reduced density operator $\rho_{A_1|A_i} = \mathrm{tr}_{A_1 \cdots A_{i-1} A_{i+1} \cdots A_k}(\rho_{A_1|A_2 \cdots A_k})$ of the joint subsystems $A_1$ and $A_i$ for $i = 2, \dots, k$.
For the entanglement measure $\mathcal{C}^{t}_q$, the following result holds.

\begin{theorem}\label{thmhqx2}
For a multi-qubit state $\rho_{A_1 \cdots A_k}$, and for any value of $2 \leq q \leq 3$, the following inequality holds
\begin{equation}
\mathcal{C}^{t}_q\left( \rho_{A_1|A_2\cdots A_k}\right) \geq \sum_{i=2}^k \mathcal{C}^{t}_q(\rho_{A_1 A_i}), \label{Tmono}
\end{equation}
where $\mathcal{C}^{t}_q\left( \rho_{A_1|A_2\cdots A_k}\right)$ represents the $\mathcal{C}^{t}_q$-concurrence of the state $\rho_{A_1A_2\cdots A_k}$ with respect to the bipartition between subsystem $A_1$ and the rest of the system, $A_2\cdots A_k$. On the other hand, $\mathcal{C}^{t}_q(\rho_{A_1 A_i})$ denotes the $\mathcal{C}^{t}_q$-concurrence of the reduced density matrix $\rho_{A_1 A_i}$, which is obtained by tracing out all subsystems except for $A_1$ and $A_i$, where $i$ ranges from $2$ to $k$.
\end{theorem}

\begin{proof}
For $q = 2$ or $3$,  (\ref{thmhqx1}) leads to the relation
\begin{equation}
\mathcal{C}^{t}_2\left(\rho_{AB}\right)=C^2(\rho_{AB})=\mathcal{C}^{t}_3\left(\rho_{AB}\right)
\label{23}
\end{equation}
which holds for any  mixed state or pure state $\rho_{AB}$ of the form ${\cal C}^2 \otimes {\cal C}^d$. Consequently, the monogamy inequality given by (\ref{Tmono}) directly follows from both (\ref{squaredconcurrence}) and (\ref{23}).

For $2<q<3$, we begin by proving the theorem for the case of an $k$-qubit pure state $\ket{\psi}_{A_1\cdots A_k}$. To begin, observe that  (\ref{squaredconcurrence}) is equivalent to the following inequality for any $k$-qubit pure state $\ket{\psi}_{A_1\cdots A_k}$,
\begin{equation}
{C}(\rho_{A_1 |A_2 \cdots A_k}) \geq  \sqrt{{C}^2(\rho_{A_1
A_2}) +\cdots+{C}^2(\rho_{A_1 A_k})}
\label{nCmonoroot}
\end{equation}
where $ C(\rho_{A_1 | A_2 \cdots A_k})$ denotes the concurrence of the state under the bipartition $A_1$ and $A_2 \cdots A_k$, and $C(\rho_{A_1 A_i})$ is the concurrence of the reduced density matrix $\rho_{A_1 A_i}$ for $ i = 2, \cdots, k$.

It has been shown that for any $2\leq q\leq 3$ \cite{Kim062328},
\begin{align}
h_q\left(\sqrt{\alpha^2+\beta^2}\right) \geq h_q(\alpha)+h_q(\beta),\label{Lemmahqx1}
\end{align}
with equality holding for $q = 2$ or $q = 3$, where $0 \leq \alpha, \beta \leq 1$ and $0 \leq \alpha^2+\beta^2 \leq 1$. Thus, it follows from (\ref{Lemmahqx1}) and (\ref{nCmonoroot}) that
\begin{align}
\mathcal{C}^{t}_q\left(\ket{\psi}_{A_1|A_2\cdots A_k} \right)=&
h_{q}\left({C}(\rho_{A_1 |A_2 \cdots A_k})\right)\nonumber\\
\geq&h_{q}\left(\sqrt{{C}^2(\rho_{A_1A_2}) +\cdots+{C}^2(\rho_{A_1 A_k})}\right)\nonumber\\
\geq& h_{q}\left({C}(\rho_{A_1A_2})\right)\nonumber\\
&+h_{q}\left(\sqrt{{C}^2(\rho_{A_1 A_3})+\cdots+{C}^2(\rho_{A_1 A_k})}\right)\nonumber\\
&~~~~~~~\vdots\nonumber\\
\geq& h_{q}\left({C}(\rho_{A_1 A_2})\right)+\cdots+h_{q}\left({C}(\rho_{A_1 A_k})\right)\nonumber\\
=& \mathcal{C}^{t}_q\left(\rho_{A_1A_2}\right)+\cdots +\mathcal{C}^{t}_q\left(\rho_{A_1A_k}\right), \label{monoineq}
\end{align}
where the first equality arises from the functional relationship between concurrence and the $\mathcal{C}^{t}_q$-concurrence for ${\cal C}^2\otimes {\cal C}^d$ pure states. The initial inequality follows from the monotonicity property of $h_{q}(x)$, while the subsequent inequalities are obtained through iterative application of (\ref{Lemmahqx1}). The final equality is established by Theorem~\ref{thmhqx}.

%For an $n$-qubit mixed state $\rho_{A_1A_2\cdots A_n}$, let$\rho_{A_1|A_2\cdots A_n}=\sum_j p_j \ket{\psi_j}_{A_1|A_2\cdots A_n}\bra{\psi_j}$ be an optimal pure state decomposition such that $\mathcal{C}^{t}_q\left(\rho_{A_1|A_2\cdots A_n}\right)=\sum_j p_j \mathcal{C}^{t}_q\left(\ket{\psi_j}_{A_1|A_2\cdots A_n}\right)$ under bipartition $A_1$ and $A_2\cdots A_n$. Since each $\ket{\psi_j}_{A_1|A_2\cdots A_n}$ in the decomposition is an $n$-qubit pure state, we have

For an $k$-qubit mixed state $\rho_{A_1A_2\cdots A_k}$, consider the optimal pure state decomposition given by $\rho_{A_1|A_2\cdots A_k}=\sum_j p_j \ket{\psi_j}_{A_1|A_2\cdots A_k}\bra{\psi_j}$,
where the $\mathcal{C}^{t}_q\left(\rho_{A_1|A_2\cdots A_k}\right)$ can be expressed as a sum of the corresponding $\mathcal{C}^{t}_q\left(\ket{\psi_j}_{A_1|A_2\cdots A_k}\right)$ for each pure state, i.e.,
$\mathcal{C}^{t}_q\left(\rho_{A_1|A_2\cdots A_k}\right)=\sum_j p_j \mathcal{C}^{t}_q\left(\ket{\psi_j}_{A_1|A_2\cdots A_k}\right)$.
Each state $\ket{\psi_j}_{A_1|A_2\cdots A_k}$ in this decomposition is an $k$-qubit pure state, which allows us to further analyze the properties of the $\mathcal{C}^{t}_q$-concurrence,
%\begin{widetext}
\begin{align}
&\mathcal{C}^{t}_q\left(\rho_{A_1|A_2\cdots A_k}\right)=\sum_j p_j
\mathcal{C}^{t}_q\left(\ket{\psi_j}_{A_1|A_2\cdots A_k}\right)\nonumber\\
\geq&\sum_j p_j\left(\mathcal{C}^{t}_q\left(\rho^j_{A_1A_2}\right)+\cdots +\mathcal{C}^{t}_q\left(\rho^j_{A_1A_k}\right) \right)\nonumber\\
=&\sum_j p_j\mathcal{C}^{t}_q\left(\rho^j_{A_1A_2}\right)+\cdots
+\sum_j p_j\mathcal{C}^{t}_q\left(\rho^j_{A_1A_k}\right) \nonumber\\
\geq&\mathcal{C}^{t}_q\left(\rho_{A_1A_2}\right)+\cdots +\mathcal{C}^{t}_q\left(\rho_{A_1A_k}\right), \label{Tmonomixed}
\end{align}
%\end{widetext}
where $\rho^j_{A_1A_j}$ represents the reduced density matrix of the state $\ket{\psi_j}_{A_1|A_2\cdots A_k}$ corresponding to the subsystem $A_1A_j$, where $j$ ranges from $2$ to $k$. The final inequality follows from the definition of the $\mathcal{C}^{t}_q$-concurrence applied to each reduced density matrix $\rho_{A_1A_j}$.
\end{proof}

Next, we provide an example to demonstrate that the $\mathcal{C}^{t}_q$-concurrence satisfies to the monogamy property.

\noindent\emph{\textbf{Example 2}}.
Let us now examine the 3-qubit state $|\phi\rangle_{ABC}$, which is expressed using the generalized Schmidt decomposition from \cite{Acin1560,Gao71} as follows
\begin{equation}
|\phi\rangle_{ABC}=\nu_0|000\rangle+\nu_1e^{i\varphi}|100\rangle+\nu_2|101\rangle+\nu_3|110\rangle+\nu_4|111\rangle,
\end{equation}
where $\nu_i\geq0$ for $i=0,1,\cdots,4$, and the normalization condition  $\sum\limits_{i=0}^{4}\nu_i^2=1$ holds. In this setup, the various concurrences are given by
$C_{A|BC} = 2\nu_0 \sqrt{\nu_2^2 + \nu_3^2 + \nu_4^2}, \quad C_{AB} = 2\nu_0 \nu_2, \quad C_{AC} = 2\nu_0 \nu_3.$
By choosing specific values for the parameters $\nu_0=\nu_3=\nu_4=\sqrt{\frac{2}{7}}$, $\nu_2=\sqrt{\frac{1}{7}}$, and $C_{AC}=\frac{4}{7}$, we obtain the following concurrences
$$
C_{A|BC} = \frac{2\sqrt{10}}{7}, \quad C_{AB} = \frac{2\sqrt{2}}{7}, \quad C_{AC} = \frac{4}{7}.
$$

For any $2\leq q\leq 3$ and $1\leq\alpha\leq4$, we have
\begin{widetext}
\begin{equation}
K_1=({\mathcal{C}^{t}_{q(A|BC)}})^\alpha=\left(\frac{1}{1-2^{1-q}}
\left(1-\left(\frac{1-\sqrt{1-(\frac{2\sqrt{10}}{7})^2}}{2} \right)^q-\left(\frac{1+\sqrt{1-(\frac{2\sqrt{10}}{7})^2}}{2} \right)^q\right)\right)^\alpha,\label{Example21}
\end{equation}
\begin{eqnarray}
K_2=({\mathcal{C}^{t}_{q(AB)}})^\alpha+({\mathcal{C}^{t}_{q(AC)}})^\alpha
=\left(\frac{1}{1-2^{1-q}}\left(1-\left(\frac{1-\sqrt{1-(\frac{2\sqrt{2}}{7})^2}}{2} \right)^q-\left(\frac{1+\sqrt{1-(\frac{2\sqrt{2}}{7})^2}}{2} \right)^q\right)\right)^\alpha\nonumber\\
+\left(\frac{1}{1-2^{1-q}}\left(1-\left(\frac{1-\sqrt{1-(\frac{4}{7})^2}}{2} \right)^q-\left(\frac{1+\sqrt{1-(\frac{4}{7})^2}}{2} \right)^q\right)\right)^\alpha.\label{Example22}
\end{eqnarray}
\end{widetext}
It is seen that $\mathcal{C}^{t}_q$-concurrence satisfies the monogamy relation for $2\leq q\leq 3$ and $1\leq\alpha\leq4$, see FIG. \ref{monoa}. FIG. \ref{monob}shows the case of $\alpha=2.5$.  FIG. \ref{monoc} shows the difference $K=K_1-K_2$  of the $\mathcal{C}^{t}_q$-concurrence  between  (\ref{Example21}) and (\ref{Example22}).

\begin{figure}[htbp]
    \begin{minipage}[t]{0.9\linewidth}
        \centering
        \includegraphics[width=\textwidth]{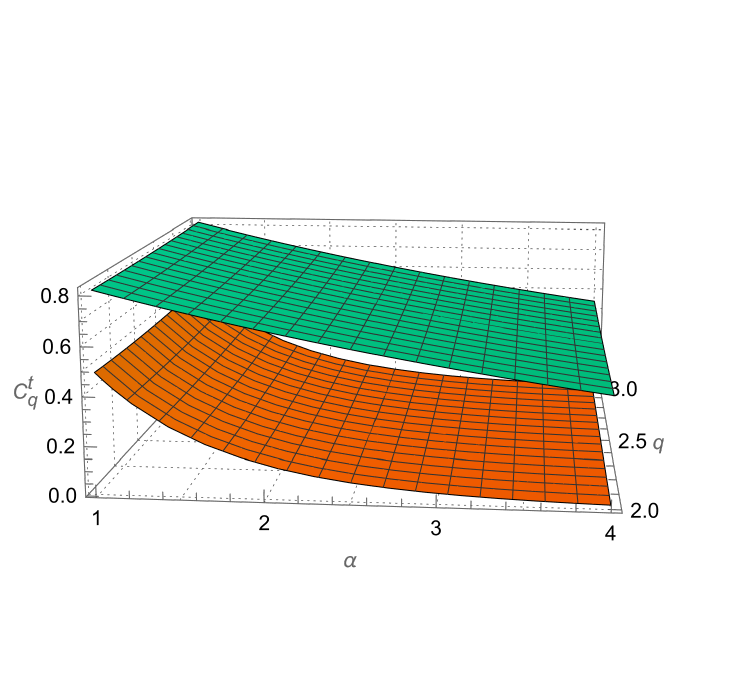}
       \caption{The green surface represents $({\mathcal{C}^{t}_{q(A|BC)}})^\alpha$, the orange surface below the green one represents $({\mathcal{C}^{t}_{q(AB)}})^\alpha+({\mathcal{C}^{t}_{q(AC)}})^\alpha$. }
    \label{monoa}
    \end{minipage}%
    \hspace{.20in}
    \begin{minipage}[t]{0.9\linewidth}
        \centering
        \includegraphics[width=\textwidth]{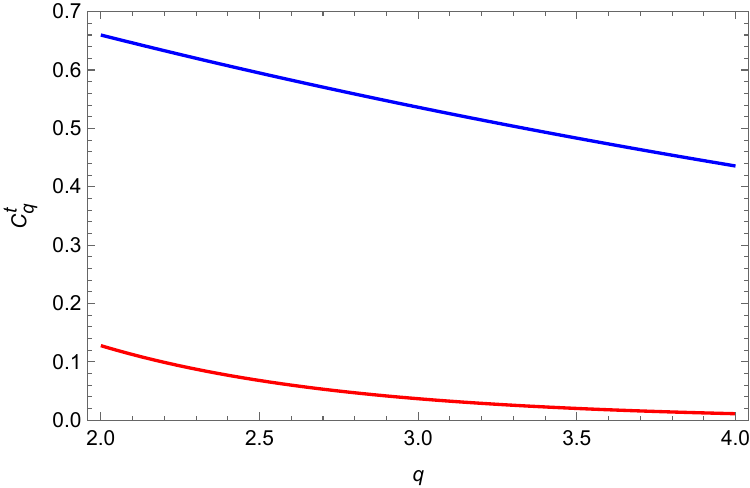}
        \caption{ The blue line represents $({\mathcal{C}^{t}_{2.5(A|BC)}})^\alpha$  and the red line represents $({\mathcal{C}^{t}_{2.5(AB)}})^\alpha+({\mathcal{C}^{t}_{2.5(AC)}})^\alpha$.}
    \label{monob}
    \end{minipage}%
  \hspace{.20in}
    \begin{minipage}[t]{0.9\linewidth}
        \centering
        \includegraphics[width=\textwidth]{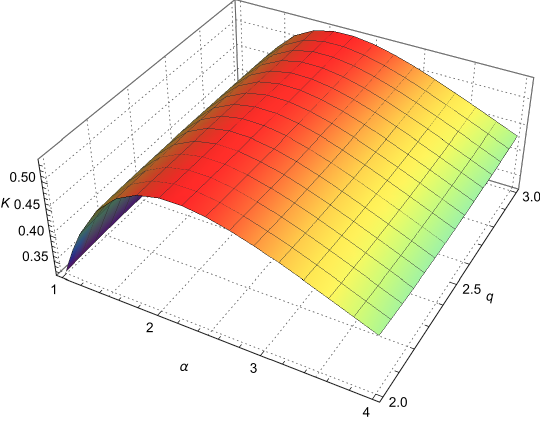}
        \caption{ The difference $K=K_1-K_2$ of the $\mathcal{C}^{t}_q$-concurrence between (\ref{Example21}) and (\ref{Example22}).}
    \end{minipage}
\label{monoc}
\end{figure}

%We have studied the monogamy property of the $\mathcal{C}^{t}_q$-concurrence.For higher-dimensional multipartite systems, the monogamy properties of the $\mathcal{C}^{t}_q$-concurrence and the concurrence are quite different.

We have explored the monogamy characteristics of the $\mathcal{C}^{t}_q$-concurrence. In higher-dimensional multipartite quantum systems, the monogamy behavior of $\mathcal{C}^{t}_q$-concurrence deviates significantly from that of standard concurrence, highlighting fundamental differences in their entanglement distribution properties. Consider a chain-network configuration composed of two EPR pairs, as introduced in \cite{Yang062402}. The corresponding quantum state in the $4 \otimes 2 \otimes 2$ system is given by
\begin{eqnarray}
|\Psi\rangle_{ABC}=\frac{1}{\sqrt{2}}(\alpha|000\rangle+\beta|110\rangle
+\alpha|201\rangle+\beta|311\rangle),\label{chain001}
\end{eqnarray}
where the parameters are defined as $\alpha = \cos\theta$ and $\beta = \sin\theta$. This state represents an entanglement structure distributed across three parties, forming a fundamental example of multipartite quantum networks. The reduced density matrix corresponding to subsystem $A$  is expressed as
\begin{eqnarray}
\varrho_A&=&\frac{\alpha^2}{2}|0\rangle\langle0|
+\frac{\beta^2}{2}|1\rangle\langle1|+\frac{\alpha^2}{2}|2\rangle\langle2|
+\frac{\beta^2}{2}|3\rangle\langle3|.\nonumber
\end{eqnarray}
From (\ref{Cq1}) we get $$\mathcal{C}^{t}_q(|\Psi\rangle_{A|BC})=\frac{4-2^{1-q} \left(\alpha ^{2 q}+\left(2-\alpha ^2\right)^q+\beta ^{2 q}+\left(2-\beta ^2\right)^q\right)}{4-4^{1-q} \left(3^q+1\right)}.$$
The reduced density matrix of the bipartite system \( AB \) is given by $\rho_{AB} = \frac{1}{2} \left( |\psi_1\rangle \langle \psi_1| + |\psi_2\rangle \langle \psi_2| \right),$
where the basis states are defined as $|\psi_1\rangle = \alpha |00\rangle + \beta |11\rangle, \quad |\psi_2\rangle = \alpha |20\rangle + \beta |31\rangle.$ Any pure state \( |\mu_i\rangle \) in an optimal pure state decomposition of \( \rho_{AB} \) can be expressed as a superposition of these basis states:
$|\mu_i\rangle = a_i |\psi_1\rangle + e^{-i\sigma} \sqrt{1-a_i^2} |\psi_2\rangle.$
The corresponding reduced density matrix of subsystem \( B \) is then given by
$\varrho^i_B = \alpha^2 |0\rangle \langle 0| + \beta^2 |1\rangle \langle 1|.$
By applying the definition of the \( \mathcal{C}^{t}_q \)-concurrence from (\ref{Cq1}), we obtain
$$\mathcal{C}^{t}_q(\rho_{AB}) = \frac{1}{1-2^{1-q}} \left( 1 - \alpha^{2q} - \beta^{2q} \right).$$
This result quantifies the entanglement between subsystems \( A \) and \( B \) in terms of the parameterized $\mathcal{C}^{t}_q$-concurrence.

Similarly for the reduced state $\rho_{AC}$, we have $\varrho^i_{C}=\frac{1}{2}|0\rangle\langle0|+\frac{1}{2}|1\rangle\langle1|$. This implies that
$\mathcal{C}^{t}_q(\rho_{AC})=1$.

Therefore, the ``residual entanglement" associated with the $\mathcal{C}^{t}_q$-concurrence can be expressed as
%Hence, the ``residual entanglement" of the $\mathcal{C}^{t}_q$-concurrence is given by
$$
 \tau_\gamma^{\mathcal{C}^{t}_q}(|\Psi\rangle_{A|BC})
 =(\mathcal{C}^{t}_q(|\Psi\rangle_{A|BC}))^\gamma-(\mathcal{C}^{t}_q(\rho_{AB}))^\gamma
 -(\mathcal{C}^{t}_q(\rho_{AC}))^\gamma.
$$
For the  concurrence, the following expressions hold
\begin{eqnarray}
&C(|\Psi\rangle_{A|BC}) = \sqrt{2 - \beta^4 - \alpha^4},\quad\\\nonumber
&C(\rho_{AB}) = \sqrt{2 - 2\beta^4 - 2\alpha^4},\quad\\\nonumber
&C(\rho_{AC}) = 1.\nonumber
\end{eqnarray}

These results characterize the entanglement distribution in the multipartite system, providing a comparison between the concurrence and the $\mathcal{C}^{t}_q$-concurrence in terms of their respective monogamy properties.

Therefore, the ``residual entanglement" of the concurrence can be expressed as
$$
\tau_\gamma^{C}(|\Psi\rangle_{A|BC})=C^\gamma(|\Psi\rangle_{A|BC})-C^\gamma(\rho_{AB})
-C^\gamma(\rho_{AC}).
$$

From FIG.\ref{Fig2} and FIG.\ref{Fig3}, it can be seen that for $\gamma\in(0,5)$, the residual entanglement of $\tau^{\mathcal{C}^{t}_q}$ and $\tau^{C}$ exhibit significantly different behaviors.
\begin{figure}[htbp]
\centering
\begin{minipage}[t]{0.9\linewidth}
\centering
\includegraphics[width=\textwidth]{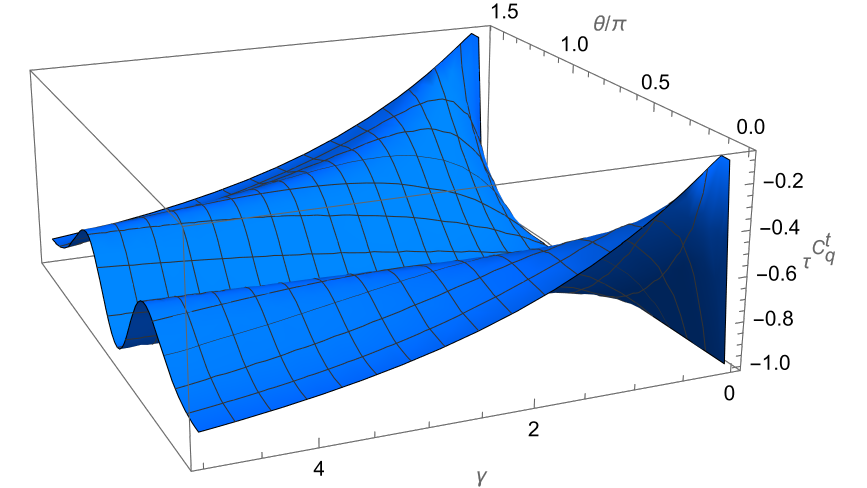}
\caption{For \( \gamma \in (0,5) \), the residual entanglement of the  $\mathcal{C}^{t}_q$-concurrence ($q=4$), is quantified by the $\tau^{\mathcal{C}^{t}_q}$.}
\label{Fig2}
\end{minipage}
\begin{minipage}[t]{0.9\linewidth}
\centering
\includegraphics[width=\textwidth]{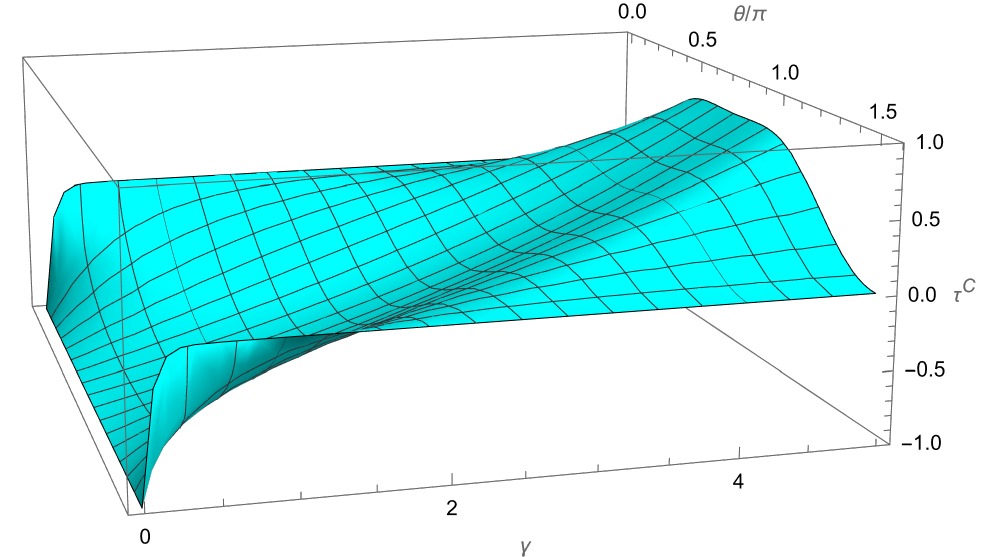}
\caption{For \( \gamma \in (0,5) \), the residual entanglement of the  concurrence, is quantified by the $\tau^{C}$.}
\label{Fig3}
\end{minipage}
\end{figure}

\section{Conclusion}\label{sect6}

%Based on the total concurrence given by the $q$-concurrence $(q\geq2)$ and its complementary dual, we have presented the new bipartite entanglement measure $\mathcal{C}^{t}_q$ $(q\geq2)$. We have derived tight lower bounds of the $\mathcal{C}^{t}_q$-concurrence for $q\geq 2$ with $d\geq 3$ and $q\geq s=3.34$ with $d=2$. Moreover, we have calculated analytically the $\mathcal{C}^{t}_q$-concurrence for isotropic states and Werner states.

Building upon the total concurrence formulated from the $q$-concurrence $(q\geq2)$ and its complementary dual, we introduced a new bipartite entanglement measure, $\mathcal{C}^{t}_q$ $(q\geq2)$. Furthermore, we established tight lower bounds for the $\mathcal{C}^{t}_q$-concurrence, valid for $q\geq 2$ when  $d\geq 3$ and for $q\geq s=3.34$ in the case of $d=2$. Additionally, we provided analytical results for the $\mathcal{C}^{t}_q$-concurrence in both isotropic and Werner states, demonstrating its applicability in characterizing entanglement in these widely studied quantum states.

We have further examined the monogamy properties of the $\mathcal{C}^{t}_q$-concurrence in both qubit and higher-dimensional systems. Compared to usual concurrence, the $\mathcal{C}^{t}_q$-concurrence reveals novel features that enable a more refined characterization of quantum entanglement. Additionally, inspired by the parameterized $\alpha$-concurrence and its complementary dual, we introduced the $\mathcal{C}^{t}_\alpha$-concurrence for $(0 \leq \alpha \leq \frac{1}{2})$. These findings open up new avenues for exploring the role of quantum entanglement in quantum information processing. Similar to the concurrence of assistance, it would be valuable to define the assisted versions of $\mathcal{C}^{t}_q$-concurrence and $\mathcal{C}^{t}_\alpha$-concurrence and investigate their polygamy properties.

%We have also investigated the monogamy property of the $\mathcal{C}^{t}_q$-concurrence for qubit systems and higher-dimensional systems. The $\mathcal{C}^{t}_q$-concurrence exhibits new characteristics compared to the usual concurrence, allowing for a more comprehensive quantification of entanglement. Moreover, based on the parameterized $\alpha$-concurrence and its complementary dual, we have proposed the $\mathcal{C}^{t}_\alpha$-concurrence $(0 \leq \alpha \leq \frac{1}{2})$. These results may highlight further investigations on applications of quantum entanglement in quantum information processing. Similar to the concurrence of assistance, it would be also interesting the define the $\mathcal{C}^{t}_q$-concurrence and $\mathcal{C}^{t}_\alpha$-concurrence of assistance and investigate their polygamy properties.
\bigskip

\begin{appendix}
\section{Entanglement measure derived from the total concurrence of parameterized \( \alpha \)-concurrence and its complementary dual}\label{a}

In \cite{WEI275303} a bipartite entanglement measure called $\alpha-$concurrence $\left(0 \leq \alpha \leq \frac{1}{2}\right)$ has been presented,
\begin{eqnarray}
\mathcal{C}_\alpha\left(|\varphi\rangle_{A B}\right)=\operatorname{tr}\left(\varrho_A^\alpha\right)-1,\label{c2}
\end{eqnarray}
here, $\varrho_A={\rm tr}_B(|\varphi\rangle_{AB}\langle\varphi|)$.

Similarly, we introduce the total concurrence associated with the parameterized $\alpha$-concurrence $(0 \leq \alpha \leq \frac{1}{2})$ and its complementary dual for a quantum state $\rho$  in a $d$-dimensional Hilbert space ${\cal H}$, defined as
\begin{eqnarray}
\mathcal{C}^{t}_\alpha(\rho)=\mathrm{tr}\rho^\alpha-1+ \mathrm{tr}(I-\rho)^\alpha -\mathrm{tr}(\textmd{I}-\rho),\nonumber
\end{eqnarray}
which aligns with the general form given in (\ref{general0}), where the function is specified as $f(x)=x^\alpha+(1-x)^\alpha-1$. In this context, the function $f(x)$ is a continuously differentiable and strictly concave real function over the interval $x\in [0,1]$, satisfying the boundary conditions $f(0)=f(1)=0$. The second derivative of $f(x)$ is given by $f''(x)=\alpha(\alpha-1)[(1-x)^{\alpha-2}+x^{\alpha-2}]<0$,
which remains negative for all $0 \leq \alpha \leq \frac{1}{2}$, confirming the concavity of $f(x)$ within the specified range.

%Here, the function $f(x)$ is a smooth and strictly concave real function in $x\in [0,1]$ satisfying $f(0)=f(1)=0$. For the second derivative of $f(x)$, we have$f''(x)=\alpha(\alpha-1)[(1-x)^{\alpha-2}+x^{\alpha-2}]<0$for any $0 \leq \alpha \leq \frac{1}{2}$.

For a pure bipartite quantum state $|\Phi\rangle_{AB}$ in the Hilbert space ${\cal H}_A\otimes {\cal H}_{B}$, we define its $\mathcal{C}^{t}_\alpha$-concurrence as
\begin{eqnarray}
\mathcal{C}^{t}_\alpha(|\Phi\rangle_{AB})=\mathrm{tr}{\rho_A}^\alpha-1+ \mathrm{tr}(I-{\rho_A})^\alpha -\mathrm{tr}(\textmd{I}-{\rho_A}),
\end{eqnarray}
where $\rho_A={\rm tr}_B(|\Phi\rangle_{AB}\langle\Phi|)$.

For a bipartite mixed state $\rho_{AB}$, we define its $\mathcal{C}^{t}_\alpha$-concurrence  via convex-roof extension,
\begin{eqnarray}
\mathcal{C}^{t}_\alpha(\rho_{AB})=\inf_{\{p_i,|\Phi_i\rangle\}}
\sum_ip_i\mathcal{C}^{t}_\alpha(|\Phi_i\rangle_{AB}),
\label{Cq222}
\end{eqnarray}
where the infimum is taken over all possible decompositions of $\rho_{AB}=\sum_ip_i|\Phi_i\rangle_{AB}\langle\Phi_i|$ into pure states.

This definition ensures that $\mathcal{C}^{t}_\alpha(\rho_{AB})$ provides a proper quantification of entanglement for general bipartite quantum states.
With regard to the $\mathcal{C}^{t}_q$-concurrence for $(q\geq 2)$, the following corollary holds:

\begin{corollary}\label{theorem5}
For any finite-dimensional quantum state $\rho_{AB}\in {\cal H}_A\otimes {\cal H}_{B}$, the $\mathcal{C}^{t}_\alpha$-concurrence $(0 \leq \alpha \leq \frac{1}{2})$  defined in (\ref{Cq222}) serves as a well-defined parameterized measure of bipartite entanglement.
\end{corollary}

\section{ Analytical expression of $\zeta\left(\varrho_\mathcal{F}\right)$ for isotropic states}\label{b}
Let $\mathcal{T}_{iso}$ denote the $\left(U\otimes U^*\right)$-twirling operator, which is defined as
$\mathcal{T}_{iso}\left(\varrho\right)=\int dU\left(U\otimes U^{\ast}\right)\varrho\left(U\otimes U^{\ast}\right)^{\dagger}$, where $dU$  represents the standard Haar measure over the group of all $d\times d$ unitary operators. Consequently, we have $\mathcal{T}_{iso}\left(\varrho\right)=\varrho_{\mathcal{F}\left(\varrho\right)}$ , where $\mathcal{F}\left(\varrho\right)=\braket{\Psi|\varrho|\Psi}$. It is also known that $\mathcal{T}_{iso}\left(\varrho_\mathcal{F}\right)=\varrho_\mathcal{F}$, as discussed in \cite{Vollbrecht062307, Horodecki4206, Lee062304}. By applying the $\mathcal{T}_{iso}$ operator to the pure state $\ket{\psi}=\sum_{k=1}^r\sqrt{\lambda_k}U_A\otimes U_B\ket{kk}$ as given in  (\ref{ac1}), where $\ket{a_k}=U_A\ket{k}$ and $\ket{b_k}=U_B\ket{k}$, we obtain the following expression
\begin{equation}\label{p4iso1}
\mathcal{T}_{iso}\left(|\psi\rangle\langle\psi|\right)
=\varrho_{\mathcal{F}\left(|\psi\rangle\langle\psi|\right)}=\varrho_{\mathcal{F}\left(\vec{\lambda},V\right)},
\end{equation}
where $V=U_A^{T}U_B$, and the function $\mathcal{F}(\vec{\lambda}, V)$ is given by
\begin{equation}\label{p4iso2}
\mathcal{F}\left(\vec{\lambda},V\right)=|\braket{\Psi|\psi}|^2=\frac{1}{d}\Big|\sum_{k=1}^r\sqrt{\lambda_k}V_{kk}\Big|^2,
\end{equation}
with $V_{kj}=\braket{k|V|j}$ and $\vec{\lambda}$ being the Schmidt vector from  (\ref{ac1}). The function $\zeta$, as defined in  (\ref{p30}), becomes
\begin{equation}\label{p4iso3}
\zeta\left(\varrho_\mathcal{F}\right)=\min_{\set{\vec{\lambda},V}}\left\{\mathcal{C}^{t}_q\left(\vec{\lambda}\right):\frac{1}{d}\Big|\sum_{k=1}^r\sqrt{\lambda_k}V_{kk}\Big|^2=\mathcal{F}\right\}.
\end{equation}
It has been demonstrated that the minimum value is achieved when $V=I$ \cite{Lee062304}. Consequently, we obtain
\begin{align}\label{p4iso8}
\zeta\left(\varrho_\mathcal{F}\right)=\min_{\vec{\lambda}}\left\{\mathcal{C}^{t}_q\left(\vec{\lambda}\right):\frac{1}{d}\Big|\sum_{k=1}^r\sqrt{\lambda_k}\Big|^2=\mathcal{F}\right\}.
\end{align}

The $\mathcal{C}^{t}_q$-concurrence for the pure state $|\psi\rangle=\sum^d_{k=1}\sqrt{\lambda_k}|a_kb_k\rangle$  is expressed in terms of the Schmidt coefficients $\lambda_k$,
\begin{eqnarray}
\mathcal{C}^{t}_q(|\psi\rangle)=d-\sum^d_{k=1}\lambda^q_k-\sum^d_{k=1}(1-\lambda_k)^q.
\label{eqnui}
\end{eqnarray}
The value of $\zeta(\mathcal{F},q,d)$ for $\mathcal{F} \in (0, \frac{1}{d}]$ can be straightforwardly computed by setting $\lambda_1 = 1$ and $v_{11} = \sqrt{\mathcal{F}}$, which gives $\zeta(\mathcal{F},q,d) = 0$. For $\mathcal{F} \in (\frac{1}{d}, 1]$, similar to \cite{Yang052423}, the optimization problem in (\ref{eqnui}) can be minimized using Lagrange multipliers \cite{Rungta012307}. The optimization is carried out subject to the following constraints
\begin{eqnarray}
&& \sum_{k} \lambda_k = 1, \\
&& \sum_{k} \sqrt{\lambda_k} = \sqrt{\mathcal{F}d},
\end{eqnarray}
with the condition that $\mathcal{F}d \geq 1$. The condition for an extremum is then given by
\begin{eqnarray}
(\sqrt{\lambda_k})^{2q-1}+\eta_1\sqrt{\lambda_k}+\eta_2=0,
\end{eqnarray}
here $\eta_1$ and $\eta_2$ are the Lagrange multipliers. It is clear that the function $f(\sqrt{\lambda_k}) = (\sqrt{\lambda_k})^{2q-1}$ is convex with respect to $\sqrt{\lambda_k}$ when $q \geq 2$.

Since a convex function and a linear function intersect at most twice, this equation has at most two nonzero solutions for $\sqrt{\lambda_k}$, which we denote as $\chi$ and $\sigma$. The Schmidt vector $\vec{\lambda} = \{\lambda_1, \lambda_2, \dots, \lambda_d\}$ then takes the following form,
\begin{eqnarray}
\lambda_l=
\left\{
\begin{aligned}
  &\chi^2, \,\,\,\,  l=1,\cdots,n &
  \\
  &\sigma^2, \,\,\,\,  l=n+1,\cdots,n+m &
  \\
  &0,        \,\,\,\, l=n+m+1,\cdots,d&
\end{aligned}
\right.
\end{eqnarray}
here $n+m\leq d$ and $n\geq1$.
The minimization problem simplifies to the following form
\begin{eqnarray}
&\min & \mathcal{C}^{t}_q(|\psi\rangle)
\label{eqminimum}
\\
&s.t.& n\chi^2+m\sigma^2=1,
\nonumber\\
& &n\chi+m\sigma=\sqrt{\mathcal{F}d},
\label{eqconstraint1}
\end{eqnarray}
where $\mathcal{C}^{t}_q(|\psi\rangle)=m+n-n(\chi^{2q}+(1-\chi^{2})^q)-m(\sigma^{2q}+(1-\sigma^{2})^q)$.

Solving (\ref{eqconstraint1}), we obtain the following solutions of $\chi$ and $\sigma$,
\begin{eqnarray}
\chi^{\pm }_{nm}(\mathcal{F})=\frac{n\sqrt{\mathcal{F}d}\pm\sqrt{nm(n+m-\mathcal{F}d)}}{n(n+m)}
\label{eqconstraint02}
\end{eqnarray}
 and
\begin{eqnarray}
\sigma^{\pm }_{nm}(\mathcal{F})
&=&\frac{m\sqrt{\mathcal{F}d}\mp\sqrt{nm(n+m-\mathcal{F}d)}}{m(n+m)}.
\label{eqconstraint2}
\end{eqnarray}
Given that $\chi^{-}_{mn} = \sigma^{+}_{nm}$, it is sufficient to focus on the solution $\chi_{nm} := \chi^{+}_{nm}$. For $\chi_{nm}$ to represent a valid solution, the expression under the square root must remain non-negative. This leads to the constraint $\mathcal{F}d \leq n + m$.

%On the other hand, $\delta_{nm}$ should be nonnegative in  (\ref{eqconstraint2}), which implies that $Fd\geq n$. In this regime, one can verify that $\delta_{nm}(F)\leq\sqrt{Fd}/(n+m)\leq\gamma_{nm}(F)$. Note that $n=0$ is not defined. Hence, we have $n\geq1$.

On the other hand, the condition in (\ref{eqconstraint2}) requires that $\sigma_{nm}$ be non-negative, which leads to the constraint $\mathcal{F}d \geq n$. Under this condition, it can be verified that $\sigma_{nm}(\mathcal{F}) \leq \frac{\sqrt{\mathcal{F}d}}{n+m} \leq \chi_{nm}(\mathcal{F})$.
It is worth emphasizing that the scenario where $n = 0$ is not defined, which necessitates $n \geq 1$.
%It is also important to note that the case $n = 0$ is not valid, which implies that we must have $n \geq 1$.
%To find the minimum of $\mathcal{C}^{t}_q(|\psi\rangle)$ over all choices of $n$ and $m$, we consider the minimization by regarding $n$ and $m$ as continuous variables, by minimizing $\mathcal{C}^{t}_q(|\psi\rangle)$ over the parallelogram defined by $1\leq n\leq Fd$ and $Fd\leq n+m\leq d$. Note that the parallelogram collapses to a line when $Fd=1$, i.e., the separability boundary. Within the parallelogram, we have $\gamma_{nm}\geq\delta_{nm}\geq0$.  $\gamma_{nm}=\delta_{nm}$ iff $n+m=Fd$ while $\delta_{nm}=0$ iff $n=Fd$.
$\min_{n, m} \mathcal{C}^{t}_q(|\psi\rangle)$

To find the $\min_{n, m} \mathcal{C}^{t}_q(|\psi\rangle)$, we approach the minimization problem by treating $n$ and $m$ as continuous variables. Specifically, we minimize $\mathcal{C}^{t}_q(|\psi\rangle)$ over the parallelogram defined by the constraints $1 \leq n \leq \mathcal{F}d$ and $\mathcal{F}d \leq n + m \leq d$. It is important to note that the parallelogram reduces to a line when $\mathcal{F}d = 1$, corresponding to the separability boundary. Within this region, we have the relations $\chi_{nm} \geq \sigma_{nm} \geq 0$. Additionally, $\chi_{nm} = \sigma_{nm}$ holds if and only if $n + m = \mathcal{F}d$, and $\sigma_{nm} = 0$ when $n = \mathcal{F}d$. By analyzing the derivatives of $\chi_{nm}$ and $\sigma_{nm}$ with respect to the variables $n$ and $m$,
\begin{eqnarray}
\nonumber&&\frac{\partial\chi}{\partial n}=\frac{1}{2n}\frac{2\chi\sigma-\chi^2}{\chi-\sigma},
\\
\nonumber &&\frac{\partial\sigma}{\partial n}=-\frac{1}{2m}\frac{\chi^2}{\chi-\sigma},
\\
\nonumber&&\frac{\partial\sigma}{\partial m}=-\frac{1}{2m}\frac{2\chi\sigma-\chi^2}{\chi-\sigma},
\\
&&\frac{\partial\chi}{\partial m}=\frac{1}{2n}\frac{\sigma^2}{\chi-\sigma},
\label{}
\end{eqnarray}
we compute the $\frac{\partial}{\partial n} \mathcal{C}^{t}_q(|\psi\rangle)$, and $\frac{\partial}{\partial m} \mathcal{C}^{t}_q(|\psi\rangle),$
\begin{eqnarray}
&\frac{\partial \mathcal{C}^{t}_q}{\partial n}=1-(\chi^{2q}+(1-\chi^{2})^{q})\\ \nonumber
&-\frac{q\chi(2\chi\sigma-\chi^{2})
(\chi^{2q-2}-(1-\chi^2)^{q-1})}{\chi-\sigma}\\ \nonumber
&+\frac{q\chi^2\sigma(\sigma^{2q-2}-(1-\sigma^2)^{q-1})}{\chi-\sigma}
\label{eqderivative1}
\end{eqnarray}
and
\begin{eqnarray}
\frac{\partial \mathcal{C}^{t}_q}{\partial m}\nonumber
&=&1-(\sigma^{2q}+(1-\sigma^{2})^{q})+q\sigma^2(\sigma^{2q-2}+(1-\sigma^{2})^{q-1}) \\ \nonumber
&-&\frac{q\chi\sigma^2(\chi^{2q-2}-\sigma^{2q-2}+(1-\sigma^2)^{q-1}-(1-\chi^2)^{q-1})}{\chi-\sigma}\\ \nonumber
&=&1-((q-1)\sigma^{2}+1)(1-\sigma^{2})^{q-1}+(q-1)\sigma^{2q}\\ \nonumber
&-&\frac{q\chi\sigma^2(\chi^{2q-2}-\sigma^{2q-2}+(1-\sigma^2)^{q-1}-(1-\chi^2)^{q-1})}{\chi-\sigma}\\
&\leq&\sigma^{4}+(q-1)\sigma^{2q}-2q\sigma^2\chi(\chi+\sigma)
\label{eqderivative21}\\
&\leq&\sigma^{4}+(q-1)\sigma^{4}-4q\sigma^{4}
\label{eqderivative22}\\
&\leq&(1+q-1-4q)\sigma^{4}
\label{eqderivative23}\\
&\leq& 0,
\label{eqderivative2}
\end{eqnarray}
where the inequality in (\ref{eqderivative21}) is validated due to the fact that the function $f(q)=1-((q-1)\sigma^{2}+1)(1-\sigma^{2})^{q-1}$ is a decreasing function of $q$. Specifically, it decreases as $q$ increases, which ensures that the inequality holds,
\begin{widetext}
\begin{eqnarray}
 \frac{\partial f}{\partial q}=-[\sigma^{2}(1-\sigma^{2})^{q-1}+(q-1)\sigma^{2}+1)(q-1)(1-\sigma^{2})^{q-2}]\leq0
 \label{eqnderivativef}
\end{eqnarray}
\end{widetext}
for $q\geq2$,
and  the  $g(q)=\frac{q\chi\sigma^2(\chi^{2q-2}-\sigma^{2q-2}
+(1-\sigma^2)^{q-1}-(1-\chi^2)^{q-1})}{\chi-\sigma}$ is an increasing function of $q$, to confirm this, we compute the partial derivative of $g(q)$ with respect to $q$, which yields,
\begin{widetext}
\begin{eqnarray}
 \frac{\partial g}{\partial q}=\frac{(2q-2)(\chi^{2q-3}-\sigma^{2q-3})+(q-1)
 ((1-\sigma^2)^{q-2}-(1-\chi^2)^{q-2}))}{\chi-\sigma}\geq0,
 \label{eqnderivativef}
\end{eqnarray}
\end{widetext}
this derivative is non-negative for $q \geq 2$ and $\chi \geq \sigma$, confirming that $g(q)$ is indeed an increasing  of $q$. And (\ref{eqderivative22}) is valid when $\chi\geq\sigma$, and the function $\nu(\sigma)=\sigma^{2q}$ decreases as $q$ increases for $q\geq2$.

Let $u=m-n$  and $v=m+n$, here $u$ and $v$ represent the motions parallel and perpendicular to the boundaries of the parallelogram defined by $m + n = c$ (with $c$ being a constant). The $\frac{\partial \mathcal{C}^{t}_q}{\partial u}$ is then expressed as
\begin{widetext}
\begin{align}
\frac{\partial \mathcal{C}^{t}_q}{\partial u}\nonumber
=&\frac{1}{2}\left(\frac{\partial \mathcal{C}^{t}_q}{\partial m}-\frac{\partial \mathcal{C}^{t}_q}{\partial n}\right)\\ \nonumber
=&\frac{(q-1)}{2}(\sigma^{2q}-\chi^{2q})+\frac{1}{2}\left((1-\chi^{2}+q\chi^{2})(1-\chi^{2})^{q-1}-(1-\sigma^{2}+q\sigma^{2})(1-\sigma^{2})^{q-1}\right)\\ \nonumber
&-\frac{q(\sigma^{2}\chi-\chi^{2}\sigma)(\chi^{2q-2}-\sigma^{2q-2}+(1-\sigma^2)^{q-1}-(1-\chi^2)^{q-1})}
{2(\chi-\sigma)}\\
\leq&\frac{(q-1)}{2}(\sigma^{4}-\chi^{4})+q\chi\sigma(\chi^2-\sigma^2)\label{eqderivative30}\\
\leq&\frac{1}{2}(\sigma^{4}-\chi^{4})+2\chi\sigma(\chi^2-\sigma^2)\label{eqderivative32}\\
\leq&-(\chi^{2}-\sigma^{2})\frac{(\chi^2+\sigma^2+4\chi\sigma)}{2}\label{eqderivative33}\leq0,
\end{align}
\end{widetext}
here the  inequality in (\ref{eqderivative30}) is verified due to the fact that $w=(\sigma^{2q}-\chi^{2q})$  is a decreasing  for $q\geq2$, meaning that
\begin{eqnarray}
\frac{\partial w}{\partial q}=2q(\sigma^{2q-1}-\chi^{2q-1})\leq0
\end{eqnarray}
and  $y=(1-\sigma^{2}+q\sigma^{2})(1-\sigma^{2})^{q-1}$ is an decreasing function of $\sigma$, i.e.,
\begin{eqnarray}
 \frac{\partial y}{\partial \sigma}=-2q(q-1)\sigma^{3}(1-\sigma^{2})^{q-2}\leq0
\end{eqnarray}
for $q\geq2$. consequently,
\small{
$$\frac{(1-\chi^{2}+q\chi^{2})(1-\chi^{2})^{q-1}
-(1-\sigma^{2}+q\sigma^{2})(1-\sigma^{2})^{q-1}}{2}<0,$$ for $\chi\geq\sigma$.}

%and that(\ref{eqnderivativef}).
Denote $s=\frac{(q-1)}{2}(\sigma^{4}-\chi^{4}) +q\chi\sigma(\chi^2-\sigma^2)$. We can get
\begin{eqnarray}
\frac{\partial s}{\partial q}=\frac{-(\chi^{2}-\sigma^{2})(\chi^2+\sigma^2+2\chi\sigma)}{2}\leq0.
\end{eqnarray}
%Thus $s$ is a decreasing function of $q\geq2$ for $\gamma\geq\delta$. The inequality (\ref{eqderivative32}) is achieved. The inequality (\ref{eqderivative33}) holds for  $\gamma\geq\delta$.
Consequently, the function $s$ exhibits a decreasing trend for $q\geq2$ when $\chi\geq\sigma$, ensuring that the inequality (\ref{eqderivative32}) is satisfied. Similarly, the validity of the inequality (\ref{eqderivative33}) is guaranteed under the condition $\chi\geq\sigma$.

%From Eqs. (\ref{eqderivative2}) and  (\ref{eqderivative33}), it is obvious that $\frac{\partial \mathcal{C}^{t}_q}{\partial m}\leq0$ within the parallelogram, and $\frac{\partial \mathcal{C}^{t}_q}{\partial u}\leq0$ except on the boundary $m+n=Fd$, where it is zero. These results imply that the minimum of  $C_q(|\psi\rangle)$ occurs at the vertex of $n=1$ and $m=d-1$.

By Eqs. (\ref{eqderivative2}) and (\ref{eqderivative33}), it is evident that within the parallelogram, the partial derivative $\frac{\partial \mathcal{C}^{t}_q}{\partial m}\leq0$, while  $\frac{\partial \mathcal{C}^{t}_q}{\partial u}\leq0$ everywhere except along the boundary where $m+n=\mathcal{F}d$, where it vanishes. These results suggest that the minimum of $\mathcal{C}^{t}_q(|\psi\rangle)$ occurs at the vertex where $n=1$ and $m=d-1$.

Consequently, the minimum value of $\mathcal{C}^{t}_q(|\psi\rangle)$ is given by
\begin{eqnarray}
\mathcal{C}^{t}_q(|\psi\rangle)&=&d-(\chi^{2q}_{1,d-1}+(1-\chi^{2}_{1,d-1})^q)\\ \nonumber
&-&(d-1)(\sigma^{2q}_{1,d-1}+(1-\sigma^{2}_{1,d-1})^q).
\end{eqnarray}
Following this approach, we obtain an explicit analytical expression that
\begin{eqnarray}
\zeta(\mathcal{F},q,d)&=&d-(\chi^{2q}+(1-\chi^{2})^{q}))\\ \nonumber
&-&\left(d-1\right(\sigma^{2q}+(1-\sigma^{2})^{q}),
\label{eqnfqd}
\end{eqnarray}
where the parameters $\chi$ and $\sigma$ are determined as follows
\begin{eqnarray}
\nonumber&&\chi=\frac{1}{\sqrt{d}}(\sqrt{\mathcal{F}}+\sqrt{(d-1)(1-\mathcal{F})}), \\&&\sigma=\frac{1}{\sqrt{d}}(\sqrt{\mathcal{F}}-\frac{\sqrt{1-\mathcal{F}}}{\sqrt{d-1}}).
\label{}
\end{eqnarray}
%Thus, the $\mathcal{C}^{t}_q$-concurrence for isotropic states is $\mathcal{C}^{t}_q(\rho_F)=co(\zeta(F,q,d))$, with $\zeta(F,q,d)$ given by (\ref{eqnfqd}).
Consequently, the $\mathcal{C}^{t}_q$-concurrence for isotropic states is expressed as
$\mathcal{C}^{t}_q(\varrho_F)=co(\zeta(\mathcal{F},q,d))$, where $\zeta(\mathcal{F},q,d)$ is defined in  (\ref{eqnfqd}).

\section{ Analytical expression of $\zeta\left(\varrho_W\right)$  for Werner states}\label{c}

%Let $\mathcal{T}_{wer}\left(\rho\right)=\int dU\left(U\otimes U\right)\rho\left(U^{\dagger}\otimes U^{\dagger}\right)$ be the $\left(U\otimes U\right)$-twirling transformations \cite{PhysRevA.64.062307}. Then $\mathcal{T}_{wer}\left(\rho\right)=\rho_{W\left(\rho\right)}$, where $W\left(\rho\right)=\mathrm{tr}\left(\rho\sum_{i<j}|\Psi_{ij}^-\rangle\langle\Psi_{ij}^-|\right)$ and $\mathcal{T}_{wer}\left(\rho_W\right)=\rho_W$ \cite{PhysRevA.59.4206,PhysRevA.68.062304}. Applying $\mathcal{T}_{wer}$ to the pure state $\ket{\psi}$ defined in (\ref{ac1}), $\ket{\psi}=\sum_{i=1}^r\sqrt{\lambda_i}U_A\otimes U_B\ket{ii}$, we have

Define the $\left(U\otimes U\right)$-twirling transformation as
$\mathcal{T}_{wer}\left(\varrho\right)=\int dU\left(U\otimes U\right)\varrho\left(U^{\dagger}\otimes U^{\dagger}\right)$
which maps any density matrix $\varrho$ to a Werner state $\varrho_\textit{w}$ characterized by
$\mathcal{T}_{wer}\left(\varrho\right)=\varrho_{\textit{w}\left(\varrho\right)}$
where the Werner parameter is given by
$\textit{w}\left(\varrho\right)=\mathrm{tr}\left(\varrho\sum_{l<k}|\Phi_{lk}^-\rangle\langle\Phi_{lk}^-|\right)$.
Moreover, the $\varrho_\textit{w}$ remains invariant under this transformation, i.e.,
$\mathcal{T}_{wer}\left(\varrho_\textit{w}\right)=\varrho_\textit{w}$ \cite{Horodecki4206,Lee062304}.
Applying the twirling operation $\mathcal{T}_{wer}$ to the pure state $\ket{\psi}$ in (\ref{ac1}), where
$\ket{\psi}=\sum_{l=1}^r\sqrt{\lambda_l}U_A\otimes U_B\ket{ll}$,
we obtain
\begin{equation}\label{ap3w1}
\mathcal{T}_{wer}\left(|\psi\rangle\langle\psi|\right)=\varrho_{\textit{w}\left(|\psi\rangle\langle\psi|\right)}=\varrho_{\textit{w}\left(\vec{\lambda},\Lambda\right)},
\end{equation}
where $\Lambda=U_A^{\dagger}U_B$  and the Werner parameter $\textit{w}(\vec{\lambda}, \Lambda)$ is given by
\begin{align}\label{ap3w2}
\textit{w}\left(\vec{\lambda},\Lambda\right)&=\sum_{l<k}|\braket{\Psi_{lk}^-|\psi}|^2\nonumber\\
&=\frac{1}{2}\sum_{l<k}|\sqrt{\lambda_l}\Lambda_{kl}-\sqrt{\lambda_k}\Lambda_{lk}|^2,
\end{align}
where $\Lambda_{lk}=\braket{l|\Lambda|k}$. Then the function $\zeta$ defined in (\ref{p30})  can be expressed as
\begin{equation}\label{ap3w3}
\zeta\left(\varrho_\textit{w}\right)=\min_{\set{\vec{\lambda},\Lambda}}\left\{\mathcal{C}^{t}_q\left(\vec{\lambda}\right):\textit{w}\left(\vec{\lambda},\Lambda\right)=\textit{w}\right\}.
\end{equation}
By using $\textit{w}\left(\vec{\lambda},\Lambda\right)=\textit{w}$, we obtain
\begin{align}\label{ap3w4}
2\textit{w}&=1-\sum_{l=1}^r\lambda_l|\Lambda_{ll}|^2-2\sum_{l<k}\sqrt{\lambda_l\lambda_k}\mathrm{Re}\left(\Lambda_{lk}\Lambda_{kl}^{\ast}\right)\nonumber\\
&\leq1+2\sum_{l<k}\sqrt{\lambda_l\lambda_k}|\mathrm{Re}\left(\Lambda_{lk}\Lambda_{kl}^{\ast}\right)|\nonumber\\
&\leq1+2\sum_{l<k}\sqrt{\lambda_l\lambda_k}\nonumber\\
&=|\sum_{l=1}^r\sqrt{\lambda_l}|^2,
\end{align}
where $\mathrm{Re}\left(z\right)$ denotes the real part of the complex number $z$.

%Note that the equalities in (\ref{ap3w4}) hold if only the two nozero components $\Lambda_{01}=1$ and $\Lambda_{10}=-1$, and $\vec{\lambda}=\left(\lambda_1,\lambda_2,0,\cdots,0\right)$, which give rise to the optimal minimum of (\ref{ap3w3}) \cite{PhysRevA.64.062307}.

It is important to note that the equalities in (\ref{ap3w4}) are satisfied only when the two nonzero components $\Lambda_{01}=1$ and $\Lambda_{10}=-1$, and $\vec{\lambda}=\left(\lambda_1,\lambda_2,0,\cdots,0\right)$, which together yield the optimal minimum of (\ref{ap3w3}) \cite{Vollbrecht062307}. Thus,  (\ref{ap3w3}) simplifies to
\begin{equation}\label{ap3w5}
\zeta\left(\varrho_\textit{w}\right)=\min_{\vec{\lambda}}\left\{\mathcal{C}^{t}_q\left(\vec{\lambda}\right):|\sum_{l=1}^2\sqrt{\lambda_l}|^2=2\textit{w}\right\}.
\end{equation}
For values of $\textit{w}\in(0,\frac{1}{2}]$, it is always possible to choose appropriate unitary transformations $V_A$ and $V_B$, such that $\lambda_1=1$, leading to $\zeta\left(\varrho_\textit{w}\right)=0$. When $\textit{w}>\frac{1}{2}$, the minimization of equation (\ref{ap3w5}) must be carried out subject to the following constraints
\begin{align}
\sum_{l=1}^2\lambda_l=1,~~~\sum_{l=1}^2\sqrt{\lambda_l}=\sqrt{2\textit{w}}.
\end{align}
The remaining calculation follows a similar approach to that in Appendix \ref{b}. By setting $d=2$ and $\mathcal{F}=\textit{w}$, we obtain the result
\begin{equation}\label{ap3w6}
\zeta\left(\varrho_\textit{w}\right)=2\left(1-\left(\frac{1+G}{2}\right)^{q}-\left(\frac{1-G}{2}\right)^{q}\right),
\end{equation}
where $G=2\sqrt{\textit{w}(1-\textit{w})}$.

\end{appendix}

\begin{acknowledgements} This work was supported by the National Natural Science Foundation of China (NSFC) under Grant 12171044, and the specific research fund of the Innovation Platform for Academicians of Hainan Province. The authors grateful to Zhi-Wei Wei and Zhi-Xiang Jin for very helpful discussions.
\end{acknowledgements}

\bigskip
%\noindent{\bf Data availability statement}
%All data generated or analysed during this study are included in this published article.


\begin{thebibliography}{99}

\bibitem{Nielsen2000book} M. A. Nielsen, I. L. Chuang, Quantum Information and Quantum Computation. 2000 Cambridge University Press, Cambridge.


\bibitem{Horodecki865}  R. Horodecki, P. Horodecki, M. Horodecki, K. Horodecki, Quantum entanglement, \href{https://doi.org/10.1103/RevModPhys.81.865}{ Rev. Mod. Phys. \textbf{81}, 865 (2009).}



 \bibitem{Vedral2275}   V. Vedral, M. B. Plenio, M. A. Rippin,  P. L. Knight,
Quantifying entanglement, \href{https://doi.org/10.1103/PhysRevLett.78.2275}{ Phys. Rev. Lett. \textbf{78} 2275 (1997).}


\bibitem{Hill5022} S. Hill,   W. K. Wootters, Entanglement of a pair of quantum bits, \href{https://doi.org/10.1103/PhysRevLett.78.5022}{ Phys. Rev. Lett. \textbf{78}, 5022 (1997).}


\bibitem{Zyczkowski883} K. \.{Z}yczkowski, P. Horodecki, A. Sanpera, and M. Lewenstein,
Volume of the set of separable states, \href{https://link.aps.org/doi/10.1103/PhysRevA.58.883}
{Phys. Rev. A \textbf{58}, 883 (1998).}



\bibitem{Bennett3824} C. H. Bennett, D. P. DiVincenzo, J. A.  Smolin, and W. K. Wootters, Mixed state entanglement and quantum error correction,
    \href{https://doi.org/10.1103/PhysRevA.54.3824} {Phys. Rev. A \textbf{54}, 3824 (1996).}



\bibitem{HORODECKI377}R. Horodecki, P. Horodecki, and M. Horodecki, Quantum $\alpha$-entropy inequalities: independent condition for local realism?
    \href{https://doi.org/10.1016/0375-9601(95)00930-2}{Phys. Lett. A \textbf{210}, 377 (1996).}


\bibitem{Gour012108} G. Gour,   S. Bandyopadhyay,  and B. C. Sanders, Dual monogamy inequality for entanglement,
    \href{https://doi.org/10.1063/1.2435088}{J. Math. Phys. \textbf{48}, 012108 (2007).}


\bibitem{LANDSBERG211} P. T.  Landsberg, and V. Vedral, Distributions and channel capacities in generalized statistical mechanics,
     \href{https://doi.org/10.1016/S0375-9601(98)00500-3}{Phys. Lett. A \textbf{247}, 211 (1998).}




\bibitem{Simon052327} C. Simon, and J. Kempe, Robustness of multiparty entanglement, \href{https://doi.org/10.1103/PhysRevA.65.052327} {Phys. Rev. A \textbf{65}, 052327 (2002).}


\bibitem{Kim295303}J. S. Kim,  and B. C. Sanders, Unified entropy, entanglement measures and monogamy of multi-party entanglement,
    \href{https://doi.org/10.1088/1751-8113/44/29/295303}{J. Phys. A Math. Theor. \textbf{44}, 295303 (2011).}


\bibitem{Yang052423} X. Yang, M. X. Luo, Y. H. Yang, and S. M. Fei, Parametrized entanglement monotone,
\href{https://doi.org/10.1103/PhysRevA.103.052423}{Phys. Rev. A \textbf{103}, 052423 (2021).}



\bibitem{WEI210}Z. W. Wei, M. X. Luo, and S. M. Fei,  Estimating parameterized entanglement measure, \href{https://doi.org/10.1007/s11128-022-03551-4}{Quantum Inf Process \textbf{21}, 210 (2022).}

\bibitem{WEI275303}Z. W. Wei and S. M. Fei, Parameterized bipartite entanglement measure, \href{https://doi.org/10.1088/1751-8121/ac7592}{J. Phys. A: Math. Theor. \textbf{55}, 275303 (2022).}


\bibitem{Rungta042315} P.  Rungta, V.  B\v{z}ek, C. M. Caves, M.  Hillery,  and G. J. Milburn,  Universal state inversion and concurrence in arbitrary dimensions,  \href{https://doi.org/10.1103/PhysRevA.64.042315}{Phys. Rev. A  \textbf{64}, 042315 (2001).}





\bibitem{COFFMAN052306} V. Coffman, J. Kundu, W.K. Wootters, Distributed entanglement, \href{https://doi.org/10.1103/PhysRevA.61.052306}{Phys. Rev. A \textbf{61}, 052306 (2000).}









\bibitem{Horodecki3} M. Horodecki, Entanglement measures,  \href{https://api.semanticscholar.org/CorpusID:4666731}{Quantum Inf. Comput. \textbf{1}, 3 (2001).}


\bibitem{Guhne1} O. G\"{u}hne, and G. T\'{o}th, Entanglement detection, \href{https://doi.org/10.1016/j.physrep.2009.02.004}{Phys. Rep. \textbf{474}, 1 (2009).}


\bibitem{Vidal355} G. Vidal, Entanglement monotones, \href{https://doi.org/10.1080/09500340008244048}{J. Mod. Optics \textbf{47}, 355 (2000).}


\bibitem{Bennett1070} C. H. Bennett, D. P. DiVincenzo, C. A. Fuchs, et al. Quantum nonlocality without entanglement,  \href{https://link.aps.org/doi/10.1103/PhysRevA.59.1070}{Phys. Rev. A \textbf{59}, 1070 (1999).}


\bibitem{Donald4252} M. J. Donald, M. Horodecki, and O. Rudolph, The uniqueness theorem for entanglement measures, \href{https://doi.org/10.1063/1.1495917}{J. Math. Phys. \textbf{43}, 4252 (2002).}


\bibitem{Canosa170401}N. Canosa and R. Rossignoli, Generalized nonadditive entropies and quantum entanglement, \href{https://doi.org/10.1103/PhysRevLett.88.170401}{Phys. Rev. Lett. \textbf{88}, 170401 (2002).}


\bibitem{Renou070403} M.-O. Renou, Y. Wang, S. Boreiri, S. Beigi, N. Gisin, N. Brunner, Limits on correlations in networks for quantum and no-signaling resources, \href{https://doi.org/10.1103/PhysRevLett.123.070403 }{Phys. Rev. Lett. \textbf{123}, 070403 (2019).}


\bibitem{Santos024101} E. Santos, M. Ferrero, Linear entropy and Bell inequalities,
 \href{https://doi.org/10.1103/PhysRevA.62.024101}{Phys. Rev. A \textbf{62}, 024101 (2000).}



\bibitem{Mintert207} F. Mintert, A. Carvalho, M. Ku\'{e}, A. Buchleitner, Measures and dynamics of entangled states,
    \href{https://doi.org/10.1016/j.physrep.2005.04.006}{Phys. Rep. \textbf{415} 207 (2005).}


\bibitem{Nielsen436} M. A. Nielsen, Conditions for a class of entanglement transformations, \href{https://link.aps.org/doi/10.1103/PhysRevLett.83.436}{Phys. Rev. Lett. \textbf{83} 436 (1999).}


\bibitem{Ando163} T. Ando, Majorization, doubly stochastic matrices, and comparison of eigenvalues,
\href{https://doi.org/10.1016/0024-3795(89)90580-6}{Linear Algebr and Appl. \textbf{118} 163 (1989).}




\bibitem{Horodecki4206} M. Horodecki, P. Horodecki, Reduction criterion of separability and limits for a class of distillation protocols,
    \href{https://doi.org/10.1103/PhysRevA.59.4206} {Phys. Rev. A. \textbf{59}, 4206 (1999).}




\bibitem{Rudolph219} O. Rudolph, Further results on the cross norm criterion for separability, \href{https://doi.org/10.1007/s11128-005-5664-1}{Quantum Inf Process. \textbf{4}, 219 (2005).}


\bibitem{Chen040504}K. Chen, S. Albeverio,  S. M.  Fei, Concurrence of arbitrary dimensional bipartite quantum states,
    \href{https://link.aps.org/doi/10.1103/PhysRevLett.95.040504}{Phys. Rev. Lett. \textbf{95}, 040504 (2005).}


\bibitem{Rungta012307} P. Rungta, C. M. Caves, Concurrence-based entanglement measures for isotropic states,
    \href{https://doi.org/10.1103/PhysRevA.67.012307}{Phys. Rev. A. \textbf{67}, 012307 (2003).}


\bibitem{Vollbrecht062307}K. G. H. Vollbrecht,  R. F. Werner, Entanglement measures under symmetry, \href{https://link.aps.org/doi/10.1103/PhysRevA.64.062307}{Phys. Rev. A \textbf{64}, 062307 (2001).}


\bibitem{Lee062304}S. Lee, D. P. Chi, S. D. Oh,  J. Kim,  Convex-roof extended negativity as an entanglement measure for bipartite quantum systems, \href{https://link.aps.org/doi/10.1103/PhysRevA.68.062304}{Phys. Rev. A \textbf{68}, 062304 (2003).}


\bibitem{Werner4277} R. F. Werner, Quantum states with Einstein-Podolsky-Rosen correlations admitting a hidden-variable model,
    \href{https://link.aps.org/doi/10.1103/PhysRevA.40.4277}{Phys. Rev. A \textbf{40}, 4277 (1989).}


\bibitem{Terhal2625} B. M. Terhal, K. G. H. Vollbrecht, Entanglement of formation for Isotropic states,
    \href{https://doi.org/10.1103/PhysRevLett.85.2625}{Phys. Rev. Lett. \textbf{85}, 2625 (2000).}


\bibitem{Wang022324}Y. X. Wang, L. Z. Mu, V. Vedral, and H. Fan,  Entanglement R\'enyi $\ensuremath{\alpha}$ entropy,
    \href{ https://doi.org/10.1103/PhysRevA.93.022324} {Phys. Rev. A. \textbf{93}, 022324 (2016).}


\bibitem{Vidal032314} G. Vidal,  R. F. Werner, Computable measure of entanglement, \href{https://doi.org/10.1103/PhysRevA.65.032314} {Phys. Rev. A. \textbf{65}, 032314 (2002).}


\bibitem{Hiroshima044302} T. Hiroshima and S. Ishizaka, Local and nonlocal properties of Werner states, \href{https://link.aps.org/doi/10.1103/PhysRevA.62.044302}{Phys. Rev. A \textbf{62}, 044302 (2000). }


\bibitem{Kim062328}J. S. Kim, Tsallis entropy and entanglement constraints in multiqubit systems,
  \href{https://link.aps.org/doi/10.1103/PhysRevA.81.062328}{Phys. Rev. A \textbf{81}, 062328 (2010).}


\bibitem{Wootters2245}W. K. Wootters, Entanglement of formation of an arbitrary state of two qubits, \href{https://link.aps.org/doi/10.1103/PhysRevLett.80.2245}{Phys. Rev. Lett. \textbf{80}, 2245 (1998).}



\bibitem{Osborne220503} T. J. Osborne and F. Verstraete, General Monogamy Inequality
for Bipartite Qubit Entanglement,
\href{https://link.aps.org/doi/10.1103/PhysRevLett.96.220503}{Phys. Rev. Lett. \textbf{96}, 220503 (2006).}


\bibitem{Acin1560}A. Ac\'{\i}n, A. Andrianov,  L. Costa, E. Jan$\acute{e}$, J. I. Latorre, R. Tarrach, Generalized Schmidt decomposition and classification of three-quantum-bit states, \href{https://link.aps.org/doi/10.1103/PhysRevLett.85.1560}{ Phys. Rev.Lett. \textbf{85}, 1560 (2000).}


\bibitem{Gao71}X. H. Gao, S. M. Fei,  Estimation of concurrence for multipartite mixed states, \href{https://doi.org/10.1140/epjst/e2008-00694-x}{ Eur. Phys. J. Spec. Topics \textbf{159}, 71 (2008).}


\bibitem{Yang062402} X. Yang, Y. H. Yang, M. X. Luo, Entanglement polygon inequality in qudit systems, \href{https://doi.org/10.1103/PhysRevA.105.062402}{Phys. Rev. A \textbf{105}, 062402 (2022).}

\end{thebibliography}
\end{document}